\documentclass[letterpaper, 10 pt, conference]{ieeeconf}  % Comment this line out
\IEEEoverridecommandlockouts                              % This command is only
\usepackage{times}
\usepackage{subcaption}
\usepackage{algorithmic}
\usepackage{url}
\usepackage{mathtools}

\usepackage{amsmath, amssymb, amscd, amsfonts}
\usepackage{graphicx}
\usepackage{hyperref}
\usepackage{enumerate}
\usepackage{xcolor}
\usepackage{float}
\usepackage[ruled,vlined]{algorithm2e}
\SetKwInput{KwInit}{Initialize}

\usepackage{caption}
\usepackage{nomencl}
\makenomenclature

\newtheorem{theorem}{Theorem}
\newtheorem{remark}{Remark}

\newtheorem{assumption}{Assumption}

\title{\LARGE \bf
Stochastic MPC with Realization-Adaptive Constraint Tightening
}

\author{Hotae Lee, Monimoy Bujarbaruah, Francesco Borrelli% <-this % stops a space
\thanks{The authors are with the Department of Mechanical Engineering, UC Berkeley, CA, USA; Email: \{hotae.lee, monimoyb, fborrelli\}@berkeley.edu}
}

\begin{document}

\maketitle
\thispagestyle{empty}
\pagestyle{empty}

%%%%%%%%%%%%%%%%%%%%%%%%%%%%%%%%%%%%%%%%%%%%%%%%%%%%%%%%%%%%%%%%%%%%%%%%%%%%%%%%
\begin{abstract}
This paper presents a stochastic model predictive controller (SMPC) for linear time-invariant systems in the presence of additive disturbances. The distribution of the disturbance is unknown and is assumed to have a bounded support. 
% The bounded support of disturbances is assumed to be known, but the distribution is unknown. 
A sample-based strategy is used to compute sets of disturbance sequences necessary for robustifying the state chance constraints. These sets are constructed \emph{offline} using samples of the disturbance extracted from its support. For \emph{online} MPC implementation, we propose a novel reformulation strategy of the chance constraints, where the constraint tightening is computed by adjusting the offline computed sets based on the previously realized disturbances along the trajectory. 

The proposed MPC is recursive feasible and can lower conservatism over existing SMPC approaches at the cost of higher offline computational time. 
Numerical simulations demonstrate the effectiveness of the proposed approach.
\end{abstract}

%%%%%%%%%%%%%%%%%%%%%%%%%%%%%%%%%%%%%%%%%%%%%%%%%%%%%%%%%%%%%%%%%%%%%%%%%%%%%%%%
\section{Introduction}
% In Model Predictive Control (MPC) approach, the presence of uncertainty in the prediction model can degrade control performance and violate constraint satisfaction. To address this issue, Robust MPC has been developed, guaranteeing constraint satisfaction for all possible disturbances.
% However, the resulting policy can be too conservative and some problems
% might allow the controller to violate the constraints within some bounded probability. It motivated the development of Stochastic MPC (SMPC) dealing with chance constraints for states. 

% \textcolor{blue}{
% Each comes with its own set of drawbacks. Sample based methods cite 123...and methods such as cite....there are no closed loop garantees, however no distrubtuo assumptipn. In order to have closed loop cite...stringlyfeasibile...conservative. Conservative why? 
% }

Stochastic model predictive control (SMPC) is a well established technique of MPC design for uncertain systems, where the state constraints are satisfied in probability with a user-specified bound~\cite{farina2016stochastic}. This allows for violations of the constraints in order to improve the controller performance in closed-loop, where the performance is measured in terms of the closed-loop cost of trajectories. A literature review for SMPC is beyond the scope of this paper. For an overview, see a recent survey on SMPC \cite{mesbah2016stochastic}, which includes references to  \cite{cannon2012stochastic,rosolia2018stochastic} (stochastic tube MPC with pre-stabilizing feedback control), \cite{kouvaritakis2013efficient,paulson2015stability,paulson2020stochastic} (affine disturbance feedback control), \cite{zhang2013stochastic,calafiore2013stochastic,schildbach2014scenario} (stochastic programming approach), etc. 
% Affine disturbance feedback policy suffers from high-dimensional decision variables in the online optimization and has difficulty of computation of uncertainty propagation.
% Stochastic programming approaches (e.g., scenario approach) has high computational burden due to the need of a large number for many hard constraints of chosen scenarios and makes less use of dynamics information. However, it can be applied to without any assumptions on the distribution of the disturbance.   
% On the other hand, stochastic tube MPC with prestabilizing feedback control policy need less computation complexity and it can consider dynamics information more to compute uncertainty propagation because the feedback gain is fixed. 

% In this paper, we focus on pre-stabilizing feedback control policy for online computation but merging the advantage of scenario approach to compute the uncertainty propagation offline for any disturbance distribution. This uncertainty propagation computation is related to the chance constraint reformulation. 
With respect to reformulating the chance constraints, SMPC approaches can be broadly categorized into two groups:  one group \cite{rosolia2018stochastic,zhang2013stochastic,ma2014stochastic,heirung2018stochastic} imposes the chance constraints for predicted states along the horizon, given the current state. 
% i.e., $\mathbb{P}(x_{t+k|t} \in \mathcal{X}|x_t)\geq 1-\alpha~$, 
% where $x_t, x_{t+k|t}$ denote the state at current time $t$, the predicted state at $t+k$, respectively and $\mathcal{X}, \alpha$ denote a feasible state set and an allowable failure probability. 
Although easy to implement, 
% and not conservative much, 
these strategies cannot ensure recursive feasibility of the MPC problem in closed-loop. 
% may not satisfy the original constraint since it is a relaxed version of the original chance constraint and moreover it does not guarantee recursive feasibility as well. 
The second group \cite{cannon2012stochastic,kouvaritakis2010explicit,korda2011strongly} imposes the chance constraints on states along the horizon for all admissible predicted states at the previous step which are reachable from the current state under any disturbances.
% i.e., $\mathbb{P}(x_{t+k+1|t} \in \mathcal{X}|x_{t+k|t})\geq 1-\alpha$, given all reachable states $x_{t+k|t}$.
% Namely, it is a sufficient condition for original chance constraint since it certainly considers a specified probability for all possible predicted states from current state. 
This guarantees recursive feasibility of the MPC problem with increased conservatism of the resulting controller~\cite{lorenzen2016constraint}.

% In this paper, we focus on a pre-stabilizing feedback control policy for online computation but merging the advantage of scenario approach to compute the uncertainty propagation offline for any disturbance distribution. 
% Distinct from these prior works, the proposed approach can be efficiently implemented, while reducing conservatism and retaining recursive feasibility. 
% In this paper we merge the benefits of both. use samples to constricy sets offline. This asically removes the distributon knowin...then to have guarantees like metods cite 123...you change the imposed constraints based on the realization...our contribtui s can be summarized as: ...
% In recent years, practical SMPC \cite{zhang2013stochastic}, \cite{ma2014stochastic}, \cite{heirung2018stochastic} has demonstrated many success for the less conservative and better performed control than Robust MPC. But, this practical SMPC using  chance constraints of predicted states given current state has no guarantee of recursive feasibility and might not satisfy the original chance constraints of the Optimal Control problem (OCP) when we apply this method with a massive amount of Monte Carlo simulations. 
% On the other hand, the approach in recursively feasible SMPC \cite{kouvaritakis2010explicit}, \cite{korda2011strongly}, \cite{cannon2012stochastic}, which can satisfy the original constraints and recursively feasibility, can be too conservative because all worst cases up to the predicted state are considered.

In this paper, we propose a novel approach to design an affine state feedback policy as a solution to an SMPC problem with recursive feasibility guarantees. Compared to existing works \cite{cannon2012stochastic}, \cite{korda2011strongly}, this approach obtains a larger region of attraction (ROA) at the cost of increased offline computation. 
% \blue{ than ``one-step" constraint based approaches in \cite{kouvaritakis2010explicit}, \cite{korda2011strongly}.} We introduce a  
% for online computation but merging the advantage of scenario approach to compute the uncertainty propagation offline for any disturbance distribution. 
Precisely, we propose a Stochastic MPC schema  which requires
sets of disturbance sequences computed offline and its utilization on adaptive constraint tightening during online MPC implementation. In particular, in the offline phase, before the control implementation, we design sets of disturbance sequences that our controller needs to be robust against, in order to satisfy the imposed chance constraints. We propagate the additive uncertainty through system dynamics and find a bound of propagated uncertainty to tightly robustify the chance constraints. Similar to scenario based methods, such as \cite{ zhang2013stochastic,calafiore2013stochastic,campi2018introduction}, this offline step utilizes samples of disturbance.  In the online phase, during control implementation, the imposed constraints in the MPC problem are determined by utilizing the aforementioned offline designed sets, and are adapted based on past disturbance realizations. This enables a realization-adaptive constraint tightening, which retains recursive feasibility, while lowering conservatism. 
%over existing approaches such as the one in \cite{cannon2012stochastic}. % These sets are approximately constructed \emph{offline }using sampled data from previous trials and adjusted \emph{online} by past realizations of current trials. When we apply the SMPC controller with this reformulation, the SMPC problem remains recursively feasible with reformulated state chance constraints and hard input constraints. The corresponding closed-loop policy satisfies original chance constraints given initial state. The resulting policy will be less conservative than existing recursively feasible SMPC.
% and it is stable in the sense of cost. 
% We consider a bounded disturbance support with a known bound, but we do not know the distribution of disturbances.
Our key contributions can be summarized as:
\begin{itemize}
    \item We propose an approach to construct subsets of the disturbance sequence support and use them to reformulate the chance constraints. These subsets are constructed offline before control implementation, using collected samples from the system trajectories.  
    % that are to be satisfied containing disturbance sequences offline from sampled data, to satisfy the original chance constraint given initial state.
    \item Utilizing the sets constructed offline, online during the receding horizon control implementation, we propose a novel reformulation of chance constraints with constraint tightening adjusted as a function of past disturbance realizations. 
    At its core, the proposed reformulation can be interpreted as approximating the multivariate integral associated to state chance constraints at step $t+1$ by using a batch formulation involving all previous states and inputs from step $0$ to $t$. 
    % This constraint tightening is computed by  disturbance sequence sets that our controller attempts to be robust to. 
    We show that this reformulated MPC problem is recursively feasible with a confidence level. As the number of offline samples increases, this reformulated MPC controller will satisfy the original chance constraints with higher confidence.
    % We prove that the solved MPC problem is recursively feasible. 
    % and provides a less conservative control policy. 
    \item We numerically compare our proposed stochastic MPC approach with the existing recursively feasible SMPC of \cite{cannon2012stochastic}.     We pick three different examples appeared in the literature. For these examples, the proposed approach obtains up to $35\%$ larger ROA and $6\%$ lower average closed-loop cost.  The proposed method requires an additional offline computation time increasing linearly with the length of the task horizon. 
    The approach \cite{cannon2012stochastic} was chosen as a representative of the classes of approaches which impose the chance constraints on states along the horizon for all admissible predicted states at the previous step which are reachable from the current state under any disturbances.  Other methods including \cite{kouvaritakis2010explicit,korda2011strongly} belong to this class. Comparison to all other methods is outside the scope of this paper.
    % This computation time complexity can be relaxed with an efficient recursive manner. In the examples, offline computation time is 2.2701 seconds in average.
\end{itemize}
\subsection{Notation}\label{sec:notation}
% For any vector $x$,
% $\max x,~ \min x$ denote an element-wise maximum and minimum of the argument vector $x$, respectively.
For any vector $w$, $w_l, w(l), w_{[l]}$  denote a random variable at time step $l$, the realization at time step $l$, a decision variable for an optimization at time step $l$, respectively. For any matrix $A$, $[A]_i$ denotes the $i^{\mathrm{th}}$ row vector.
% For any set $\hat{W}$, $\hat{W}^\mathrm{off},~\hat{W}^\mathrm{on}$ denotes an offline computed set, an online computed set, respectively.
% \nomenclature{$w_l$}{$w$ as random variable at time $l$}
% \nomenclature{$w(l)$}{$w$ as realization at time $l$}
% \nomenclature{$w_{[l]}$ }{$w$ as decision variable for optimization of time $l$}
% \nomenclature{$w_{(l)}$}{$l$-th element of the vector $w$}
% \printnomenclature

\section{Problem Formulation}\label{sec:prob_form}
We consider an uncertain linear time-invariant (LTI) system: 
\begin{align}\label{eq:sys}
    x_{t+1} = Ax_t + Bu_t + w_t,~x_0=x_S, 
\end{align}
where the system matrices $A\in \mathbb{R}^{n\times n}$ and $B \in \mathbb{R}^{n \times m}$ are known, $x_t \in \mathbb{R}^n, u_t \in \mathbb{R}^m, w_t \in \mathbb{R}^n$ denote the state, control input and disturbance at time $t$, respectively, and the additive disturbance $w_t$ is distributed according to an \emph{unknown} probability distribution function (i.e., $w_t \sim f^{\mathrm{uk}}_{w}$), on a \emph{known} bounded support $\mathbb{W}$. 
% The system is performing an iterative task. At any iteration $j$, 
Our goal is to design a controller to regulate the system from the given initial state $x_S$, satisfying the state and input constraints given by: 
\begin{subequations}
\begin{align}
    &\mathbb{P}(H x_t \leq h) \geq 1-\bar{\alpha},~\forall t\in \{0,1, \dots, T\}, \label{eq:orig_co}\\
    &H_u u_t \leq h_u,~\forall t\in \{0,1, \dots, T-1\},
\end{align}
\end{subequations}
where $H \in \mathbb{R}^{p \times n}, H_u \in \mathbb{R}^{q \times m} $, $T\gg 0$ is the task duration, and $\bar{\alpha} \in (0,1)$ is a user specified upper bound on the probability of state constraint violation at each sample of the task duration. We apply Boole's inequality to \eqref{eq:orig_co} and consider the sufficient condition of the satisfaction of individual chance constraints as:
\begin{align}\label{eq:orig_co_scalar}
    \mathbb{P}([H]_ix_t \leq h_i) \geq 1- \frac{\bar{\alpha}}{p}=1-\alpha,~\forall i =\{1,\dots,p\},
\end{align}
where $h_i$ is the $i^{\mathrm{th}}$ element of $h$. Note that our choice of the same violation probability for each $i=\{1,2,\dots, p\}$ is only for the clarity of presentation in the subsequent sections.
\subsection{Finite Time Optimal Control Problem}
% In order to solve a chance constrained control problem, 
We find feasible solutions to the following finite time optimal control problem:
% solved in the $j$-th iteration $(j=\{1,\cdots, N_\mathrm{iter}\})$:
\begin{equation}\label{eq:original_opti}
\begin{aligned}
\min_{u_0, u_{1}(\cdot), \dots, u_{T-1}(\cdot)}~ & \sum_{t=0}^{T-1} \ell(\bar{x}_{t}, u_t(\bar{x}_t))+\ell_F(\bar{x}_T)\\
\textrm{s.t.,}~~~~~~~ & x_{t+1}  = Ax_t+Bu_t(x_t)+w_t,\\
& \bar{x}_{t+1}  = A\bar{x}_t+Bu_t(\bar{x}_t),\\
&H_u u_t(x_t)  \leq h_u,~\forall w_t \in \mathbb{W},\\
&\mathbb{P}([H]_i x_{t+1} \leq  h_i) \geq  1-\alpha,\\
&\forall i \in \{1,\dots,p\},\\
&\forall t \in \{0, 1,\dots,T-1\},\\
& x_0 = \bar{x}_0 = x_S,
\end{aligned}
\end{equation}
% \begin{equation}\label{eq:original_opti}
%     \begin{array}{llclcl}
% \displaystyle \min_{u_0, u_{1}(\cdot), \dots, u_{T-1}(\cdot)} & \multicolumn{3}{l}{\sum_{t=0}^{T-1} \ell(\bar{x}_{t}, u_t(\bar{x}_t))+\ell_F(\bar{x}_T)} \\
% ~~~~~~~~\textrm{s.t.,} &  
% x_{t+1} & = & Ax_t+Bu_t(x_t)+w_t, \\
% &\displaystyle \bar{x}_{t+1} & = & A\bar{x}_t+Bu_t(\bar{x}_t), \\
% &\displaystyle \mathbb{P}(H_u u_t(x_t)&  \leq &h_u) = 1,\\
% & \multicolumn{3}{l}{\mathbb{P}(H x_{t+1} \leq  h)  \geq  1-\alpha, }  \\
% % & \displaystyle \forall w_t & \in & \mathbb{W},\\
% & \multicolumn{3}{l}{\forall t \in \{0, 1,\dots,T-1\},}\\
% &  \multicolumn{3}{l}{x_0 = \bar{x}_0 = x_S,} 
% % &  \multicolumn{3}{l}{ ~\forall j=\{1,\cdots, N_\mathrm{iter}\},} 
% \end{array}
% \end{equation}
where $\ell(\cdot,\cdot)$ denotes the stage cost, and $\ell_F(\cdot)$ denotes the final cost. Problem \eqref{eq:original_opti} is carried out over the space of feedback policies, 
$u_t(x_t)$ which map the set of feasible states, subset of $\mathbb{R}^n$, to the set of feasible inputs, subset of $\mathbb{R}^m$. Pair $\{\bar{x}_t, u_t(\bar{x}_t)\}$ denotes the nominal state and the corresponding nominal input, respectively. 
% $N_{\mathrm{iter}}$ denote the entire task horizon and maximum iteration number, respectively. 

There are three main issues in solving~\eqref{eq:original_opti}, namely:
\begin{enumerate}[(I)]
    \item A large task horizon $T \gg 0$, can result in an unpractical computational burden while solving~\eqref{eq:original_opti}.
    \item Optimizing over policies $\{u_{0},\dots,u_{T-1}(\cdot)\}$ is an infinite dimensional problem, and computationally intractable in general.
    \item The distribution $f^\mathrm{uk}_w$ of the disturbance $w_t$ is unknown. 
\end{enumerate}
To address (I) and (II), we solve~\eqref{eq:original_opti} in a receding horizon fashion, 
% We formulate a receding horizon fashion problem with the reformulated constraints satisfying original chance constraints given initial state $x(0)$.
% We do not use the conservative chance constraint conditioned on the current observation (i.e. $\mathbb{P}(H_x x_{k+t|t} \leq  h_x | x(t)  )  \geq  (1-\alpha)^{k+t} \Rightarrow \mathbb{P}(H_x x_{k+t|t} \leq  h_x) \geq (1-\alpha)^{k+t} $) , so that we can ensure recursive constraint satisfaction of the optimization problem.  
restricting ourselves to affine state feedback policies with a fixed stabilizing feedback gain $K$, i.e.,
\begin{align}\label{eq:aff_pol}
    u_t(x_t) = Kx_{t}+v_t,
\end{align}
where $v_t$ is the auxiliary input. 
\begin{assumption}[Strictly Stable]\label{assmp:strictly_stable}
Gain $K$ is chosen such that $A+BK$ is strictly stable.
\end{assumption}
Issue (III) is addressed by using a sample-based strategy. By sampling disturbance vectors propagated through system dynamics \eqref{eq:sys}, we estimate disturbance sequence sets with the following property:
if state constraints are satisfied robustly for all disturbance sequences in the sets, then the chance constraints \eqref{eq:orig_co_scalar} are satisfied for all $i =\{1, \dots, p\}$. Since our approach is sample-based, the former statement is true at the limit, i.e., with an infinite number of samples. These statements are formalized in the next section. 

\subsection{Receding Horizon Reformulation}
We consider the MPC reformulation of~\eqref{eq:original_opti} in this section, with a horizon length of $N \ll T$. At time step $t$, $x(t)$ denotes the measured $x_t$ and let $x_{k|t}$ denote the predicted state for prediction step $k\in \{t,t+1,\dots,t+N-1\}$, obtained from $x(t)$ by applying the predicted input policies $\{u_{t|t}, \dots, u_{k-1|t}(\cdot)\}$ to \eqref{eq:sys}. In \cite{cannon2012stochastic}, recursive feasibility of the SMPC is guaranteed by imposing 
\begin{align*}
\mathbb{P}([H]_i x_{k+1|t}\leq h_i|x_{k|t}) \geq 1-\alpha~~,~ \forall i =\{1,\dots,p\}, 
\end{align*}
which means that the chance constraints at time step $t$ for the predicted step $k+1$ are imposed for all the reachable states $x_{k|t}$ from the $x(t)$.
The reachable set containing $x_{k|t}$ is computed by propagating $x(t)$ through the system dynamics under any admissible disturbances $\{w_{t}, w_{t+1}, \dots, w_{k-1}\} \in \mathbb{W}^{k-t}$ and the resulting optimal MPC policy.
This condition is sufficient for the controller to satisfy \eqref{eq:orig_co_scalar} for the states in closed-loop with the MPC controller, but can be conservative, as pointed out in \cite{lorenzen2016constraint}. In order to reduce this conservatism and yet satisfy \eqref{eq:orig_co_scalar} in closed-loop with an MPC, we propose a new controller with offline computed disturbance sets and a realization-adaptive constraint tightening.
\subsubsection*{Approach Insight}
To explain the proposed approach intuitively, we introduce a simple example in this section. We consider the system with given initial state $x_0$,
\begin{align}\label{eq:simple_sys}
   x_{t+1} = x_t + u_t +w_t, 
\end{align}
with 3 possible disturbances at every time step, i.e., $w_t\in\{-1,0,1\}$. Consider the problem of finding an input policy $u_t$  over two time steps (i.e., $t=0,1$) which satisfies the following constraints:
$ x_2 \leq 1, x_2 \geq -1$ with at least $\frac{2}{3}$ probability respectively, i.e., $\mathbb{P}(x_2 \leq 1) \geq \frac{2}{3}, \mathbb{P}(x_2 \geq \text{-}1) \geq \frac{2}{3}$. 
Hard input constraints are $-1\leq u_t \leq 1$. 

We first solve the above problem using the SMPC of \cite{cannon2012stochastic}. 
As discussed before, it formulates the chance constraints for $x_2$ as follows:
\begin{subequations}\label{eq:simple_rf}
\begin{align}
    &\mathbb{P}(x_2 \leq 1|x_1) \geq \frac{2}{3}, \label{eq:simple_rf1}\\
    &\mathbb{P}(x_2 \geq -1|x_1) \geq \frac{2}{3}. \label{eq:simple_rf2}
 \end{align}
\end{subequations}
The feasible set $\mathcal{X}_1$ of all states $x_1$, where $x_2$ satisfies the constraints \eqref{eq:simple_rf} when a suitable control policy $u_1$ is applied to the system \eqref{eq:simple_sys}, is obtained as follows. Each set of disturbances is obtained for tight robustification of \eqref{eq:simple_rf1}, \eqref{eq:simple_rf2} respectively.
% \medmuskip=0mu
% \thinmuskip=0mu
% \thickmuskip=0mu
\begin{align}\label{eq:x1_range}
    &\exists u_1, x_1\!+\!u_1\!+\!w_1 \leq 1,\forall w_1 \in \{0,-1\} \!\!\iff\!\! \exists u_1, x_1\!+\!u_1 \leq 1, \nonumber\\
    &\exists u_1, x_1\!+\!u_1\!+\!w_1 \!\geq \!-1,\forall w_1 \in \{0,1\} \!\!\iff\!\! \exists u_1, x_1\!+\!u_1 \!\geq\! -1, \nonumber\\
    & \therefore -2 \leq x_1 \leq 2 ~(\because u_1 \in [-1,1]), ~~\mathcal{X}_1 = [-2,2].
\end{align}
As discussed earlier, the approach of \cite{cannon2012stochastic} imposes \eqref{eq:simple_rf} 
for all reachable $x_1$ from $x_0$.
Therefore, we find the 1-step robust controllable set $\mathcal{X}_0$ to $[-2,2]$ $\forall w_0$ (see \cite[Chapter~10]{borrelli2017predictive} for the definition) as follows:
\begin{align*}
    &\exists u_0, -2 \leq x_0+u_0+w_0 \leq 2, ~\forall w_0 = \{-1,0,1\},\\
    \iff& \exists u_0, -1 \leq x_0\!+\!u_0 \leq 1 
    \therefore -2 \leq x_0 \leq 2,  ~\mathcal{X}_0 = [-2,2].
\end{align*}
% $\mathcal{X}_0$ is the feasible set of $x_0$ for satisfaction of \eqref{eq:simple_rf}.
% \textbf{The set eq(3) is derived as follows:
% since:(proof)}
%Consider 
%the constraint %$\mathbb{P}(x_2 \in %\mathcal{X}_f) \geq %\frac{2}{3}$.
% Using the same approach for the chance constraint on $x_1$, the feasible set $\mathcal{X}_0$ of all states $x_0$, where $x_1$ satisfies the constraint $\mathbb{P}(x_1 \in \mathcal{X}) \geq \frac{2}{3}$ when a suitable $u_0(x_0)$ is  applied to the system \eqref{eq:simple_sys}, is obtained equal to $[-\infty,3]$.
% In order to find the feasible set of $x_0$ for both constraints on $x_1$ and $x_2$ to be satisfied, we take the intersection of the two sets. 
Thus the feasible set of $x_0$, $\mathcal{X}_0^{\mathrm{Exist}}$ is $[-2,2]$ when the SMPC of \cite{cannon2012stochastic} is utilized.  

% To satisfy this terminal constraint for all admissible $x_1$, a feasible set of $x_0$ is determined as next. 
% \begin{align*}
%     &\mathbb{P}(x_2 \in \mathcal{X}_f|x_1)=\mathbb{P}(-1\le x_1+u_1+w_1\le1|x_1)\ge\frac{2}{3} \\
%     &\Rightarrow x_1\in[-2,1]~\text{or}~x_1\in[-1,2]~\forall w_0
%     \Rightarrow x_0 \in [-2,2]
% \end{align*}
% In a similar way, $\mathbb{P}(x_1 \in \mathcal{X}) \geq \frac{2}{3}$ is satisfied for all to $x_0 \in [-3,3]$. 
% Because both of two constraints should be satisfied, the feasible set of $x_0$, $\mathcal{X}_0^{\mathrm{RF}}$ is $[-2,2]$ when RF-SMPC is applied. 

Consider now the proposed approach. This approach formulates the chance constraints by successively substituting $x_t+u_t+w_t$ for $x_{t+1}$ until $x_0$ appears. According to the following expression for $x_2\leq 1, x_2 \geq -1$: \begin{subequations}\label{eq:proposed_constr}
\begin{align}
    &x_2 = x_0 +u_0 +u_1+w_0+w_1\leq 1, \label{eq:proposed_constr1}\\
    &x_2 = x_0 +u_0 +u_1+w_0+w_1\geq -1, \label{eq:proposed_constr2}
\end{align}
\end{subequations}
and the feasible set of $x_0$ is defined as a set of all states $x_0$, where $u_0,u_1$ exist to robustly satisfy \eqref{eq:proposed_constr1}, \eqref{eq:proposed_constr2} for 6 disturbance sequences (i.e., $(w_0,w_1)$) out of 9 possible disturbance sequences, respectively. 
The sets of these disturbance sequences are computed offline. Here the disturbance sequences are chosen to satisfy the $\frac{2}{3}$-probability tightly, i.e., \eqref{eq:proposed_constr1}, \eqref{eq:proposed_constr2} hold for only the corresponding chosen sequences, not other remaining sequences. The sets $\hat{W}_1$ for \eqref{eq:proposed_constr1}, $ \hat{W}_{-1}$ for \eqref{eq:proposed_constr2} are obtained uniquely as:
% $\hat{W}^{\mathrm{off}}=\{ (w_0,w_1) |(-1,-1),(0,0),(0,-1),(1,-1),(-1,1),(-1,0)\}$.
\begin{align*}
    &\hat{W}_1^{}\text{=}\{ (w_0,w_1) |(\text{-}1,\text{-}1),(\text{-}1,0),(\text{-}1,1),(0,\text{-}1),(0,0),(1,\text{-}1)\}, \\
    &\hat{W}_{-1}^{}\text{=}\{ (w_0,w_1) |(1,1),(1,0),(1,\text{-}1),(0,1),(0,0),(\text{-}1,1)\}.
\end{align*}
We compute conditions for the controller to satisfy each of \eqref{eq:proposed_constr} robustly for any disturbance sequence in each set.
% and adaptively select control inputs depending on past disturbances. 
For the first set $\hat{W}_1$, there exist $u_0, u_1$ which satisfy \eqref{eq:proposed_constr1} robustly for $\forall (w_0,w_1) \in \hat{W}_1$ as below. 
% For example, consider to decide $u_1(w_0)$ when $w_0=0$. Considering $(0,-1),(0,0)$ in $\hat{W}_1^{}$, $u_1(w_0)$ should satisfy \eqref{eq:proposed_constr1} robustly for $\forall w_1 \in \{-1,0\}$. Likewise, considering $(0,1),(0,0)$ in $\hat{W}_{-1}^{}$, $u_1(w_0)$ also should satisfy \eqref{eq:proposed_constr2} robustly for $\forall w_1 \in \{1,0\}$. The feasible set $\mathcal{X}_1$ is the set of $x_1$s:
% \begin{align*}
%     &\exists u_1, x_1\!+\!u_1\!+\!w_1 \leq1~ \forall w_1 \in \{-1,0\}
%     \!\!\iff\!\!  \exists u_1, x_1\!+\!u_1 \leq 1,\\
%     &\exists u_1, x_1\!+\!u_1\!+\!w_1 \geq-1~ \forall w_1 \in \{1,0\}
%     \!\!\iff\!\!  \exists u_1, x_1\!+\!u_1 \!\geq\! -1,\\
%     &\therefore -2\leq x_1 \leq 2.
% \end{align*}
% As $w_0 = 0$, the candidate feasible set of $x_0$ is:
\begin{align*}
    &\exists (u_0, u_1), x_0+u_0+u_1+w_0+w_1\leq 1 ~\forall (w_0,w_1) \in \hat{W}_{1}, \\
    &\iff \exists (u_0, u_1),~ x_0+u_0+u_1 \leq 1.
\end{align*}
Likewise, $u_0, u_1$ also should satisfy \eqref{eq:proposed_constr2} robustly for $\forall (w_0,w_1) \in  \hat{W}_{-1}$ as below.
\begin{align*}
    &\exists (u_0, u_1), x_0+u_0+u_1+w_0+w_1\geq \!-1 ~\forall (w_0,w_1)\! \in\! \hat{W}_{-1},\\
    &\iff \exists (u_0, u_1),~ x_0+u_0+u_1 \geq -1.
\end{align*}
Considering both constraints and the hard input constraints, we find the feasible set of $x_0$, $\mathcal{X}_0^{\mathrm{Prop}}$ as:
\begin{align*}
    &\exists (u_0, u_1), -1 \leq x_0+u_0+u_1 \leq 1 \iff -3 \leq x_0 \leq 3, \\
    &\therefore \mathcal{X}^{\mathrm{Prop}}_0 =[-3,3].
\end{align*}
% Similarly considering the cases of $w_0=1$, $w_0=-1$, we obtain the corresponding initial feasible sets $\mathcal{X}^{1}_0 = [-3,3]$ and $\mathcal{X}^{-1}_0=[-3,3]$, respectively. Therefore, the feasible set of $x_0$ 
% % under the set $\hat{W}_1, \hat{W}_2$ 
% is the intersection of three sets: $\mathcal{X}_0 = [-3,3]$.
In conclusion, when the proposed approach is applied, the final feasible set $\mathcal{X}_0^{\mathrm{Prop}}$, for which the constraints  
 $\mathbb{P}(x_2 \leq 1) \geq \frac{2}{3}, \mathbb{P}(x_2 \geq -1) \geq \frac{2}{3}$
 are satisfied, is $[-3,3]$ which is a superset of $\mathcal{X}_0^{\mathrm{Exist}}$.
Since $\mathcal{X}_0^{\mathrm{Prop}}$ contains $\mathcal{X}_0^{\mathrm{Exist}}$,
our proposed approach can be less conservative intuitively while satisfying the chance constraint.

To demonstrate this difference between two approaches, Fig. \ref{fig:intuitive_smpc} pictorially illustrates one scenario with $x_0=3$. The condition $x_0=3$ makes the SMPC problem of \cite{cannon2012stochastic} infeasible, but makes the proposed SMPC problem feasible. In Fig. \ref{fig:intuitive_smpc}, each circle represents the state evolved through the system dynamics \eqref{eq:simple_sys} without control inputs. Consider only the first constraint $\mathbb{P}(x_2 \leq 1) \geq \frac{2}{3}$ here.
MPC controllers try to make the colored circles satisfy the state constraints at $t=2$, by applying inputs over two steps. According to what we discussed earlier, the existing SMPC of \cite{cannon2012stochastic} wants to satisfy the constraints for the blue circles in the left figure, and our proposed SMPC wants to satisfy the constraints for the red circles in the right figure.
Our approach adaptively chooses the disturbances which the controller should be robust against depending on past realizations, while the existing SMPC always tries to satisfy the constraints robustly for the same two-thirds disturbances for any reachable states in the previous step. In Fig. \ref{fig:intuitive_smpc} at $t=2$, the proposed SMPC tries to steer $[1,3]$ into $[-\infty,1]$ with $u_0,u_1$ to satisfy $\mathbb{P}(x_2 \leq 1) \geq \frac{2}{3}$, while the existing SMPC tries to steer $[1,4]$ into $[-\infty,1]$. Considering $-1\leq u_t\leq 1$, we can find $[1,4]$ cannot be steered into $[-\infty,1]$ within two steps, unlike $[1,3]$. Thus, $3$ is not in $\mathcal{X}_0^{\mathrm{Exist}}$ when applying the existing SMPC, whereas $3$ is in $\mathcal{X}_0^{\mathrm{Prop}}$ when applying the proposed SMPC.
% Additionally, we do Monte-Carlo simulations to check how many violations of the state constraints actually happen when two controllers are applied. To run simulations, we pick the controller to minimize the control effort i.e. $\min \sum_t u_t^2$ while satisfying all conditions we mentioned so far. In 1000 Monte-Carlo simulations, the RF-SMPC has violation rates, 0 at time 1, 0.3130 at time 2. On the other hand, the proposed SMPC has violation rates, 0.1690 at time 1, 0.3300 at time 2, which are less conservative in allowable range ($\frac{1}{3}$ violations).
% According to comparison of the size of the feasible set of $x_0$, 
% and ii) constraint violation rates between two approaches

% $\blacksquare$
\begin{figure}[h]
    \centering
    \includegraphics[width=0.89\columnwidth]{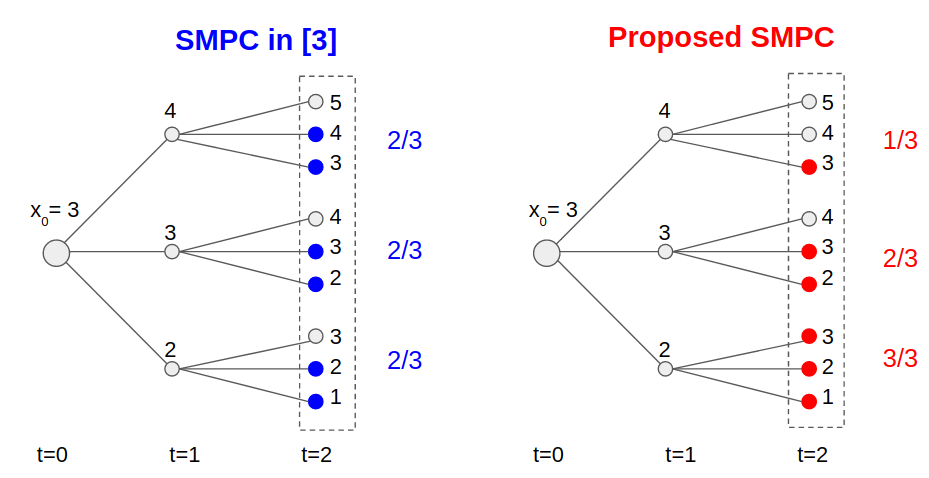}
    \caption{Tree diagram of a simple example scenario}
    \label{fig:intuitive_smpc}
\end{figure}

\noindent Next we describe the proposed approach in details:
\begin{enumerate}[A)]
\item (Offline) Firstly we define a set of all possible disturbance sequences which contain $t$ disturbances up to time step $t$ for $t=0,\dots,T-1$ as $W^{\mathrm{seq}}_{0 \rightarrow t}=\{[w_0,\dots,w_{t-1}] | ~w_k \in \mathbb{W},\forall k\}$. We find a subset $\hat{W}^{\mathrm{off},t}_{0 \rightarrow t, i}$ of this set with the following property for all $i \in \{1,2,\dots,p\}$:
If state constraints are satisfied robustly for all $[w_0,\dots,w_{t-1}]$ in the subset, then the $i^{\mathrm{th}}$-chance constraint in \eqref{eq:orig_co_scalar} is satisfied tightly.
The set consists of admissible disturbances from $0$ to $t-1$ satisfying:
% \begin{align}\label{eq:wset_prop}
$\mathbb{P}( [w_0,\dots, w_{t-1}] \in \hat{W}^{\mathrm{off},t}_{0 \rightarrow t, i} ) \geq 1-\alpha$.
% \end{align}
In other words, we should construct sets $\hat{W}^{\mathrm{off},t}_{0 \rightarrow t, i}$ which guarantee the probability that a disturbance sequence from $0$ to $t-1$ belongs to the set, is at least $1-\alpha$.
The construction of these sets is elaborated in Section~\ref{sec:estimation}. 
% If our controller satisfies the state constraints for time step $t$ robustly for all disturbances belonging to this $\hat{W}^{\mathrm{off},t}_{0 \rightarrow t,i}$ for all $i$, the chance constraints~\eqref{eq:orig_co_scalar} are also satisfied.   

 \item (Online) We design the MPC controller satisfying the state constraints robustly for any disturbance sequence in the offline computed sets. Next, let $0$ be the initial time of the control task, $t$ the current control time step and $k\geq t$ the time for which we are making predictions. Consider the realized disturbances along a trajectory until time step $t$ as $\mathbf{w}_{0 \rightarrow t-1}= [w(0),\dots, w(t-1)]$. We construct a realization-adaptive set $\hat{W}^{\mathrm{on},k+1}_{t \rightarrow k+1,i}(\mathbf{w}_{0\rightarrow t-1})$ to compute the constraint tightening for prediction step $(k+1)$'s $i^{\mathrm{th}}$-state constraint $(\forall k\in\{t,\dots,t+N-1\}, ~\forall i \in \{1,\dots,p\})$ of the MPC problem solved at $t$.
%  We design the constraint tightening for time step $(k+1)$'s chance constraint $(\forall k\in\{t,\dots,t+N-1\})$ by constructing a realization-adaptive set $\hat{W}^{\mathrm{on},k+1}_{t \rightarrow k+1}(\mathbf{w}_{0\rightarrow t-1})$ in the MPC problem solved at $t$. 
Such a set contains all the sequences $[w_t,\cdots,w_{k}]$ which the MPC controller has to be robust against, in order to satisfy the $i^{\mathrm{th}}$-state constraint for any $[w_0,\dots,w_{k}] \in \hat{W}^{\mathrm{off},k+1}_{0 \rightarrow k+1, i}$. This set is constructed using $\hat{W}^{\mathrm{off},k+1}_{0\rightarrow k+1, i}$ discussed earlier and $\mathbf{w}_{0\rightarrow t-1}$.  
The construction of the sets $\hat{W}^{\mathrm{on},k+1}_{t \rightarrow k+1, i}(\mathbf{w}_{0\rightarrow t-1})$ is elaborated in Section \ref{sec:real_dep_sets}. 
% Intuitively, if previously encountered disturbances are large, we can afford to satisfy the original chance constraint \eqref{eq:orig_co} with the smaller portion of $\mathbb{W}$ than previous one, and vice versa. 
% If previous disturbances are small, we have some slacks satisfying with larger portion of $\mathbb{W}$ than previous one. 
These disturbance realization-adaptive constraint tightenings allow the MPC problem not to impose excessively strict constraints,
while satisfying \eqref{eq:orig_co_scalar} in closed-loop.
\end{enumerate}

When both the offline and online methods are implemented, we obtain a MPC with lower conservatism over existing approaches and with recursive feasibility guarantees (in probability) at the cost of higher offline computation time.
The robust MPC reformulation of \eqref{eq:original_opti} is given by:
\begin{subequations}\label{eq:robust_reformulation_mpc}
\begin{align}
% &V_t(x_t, \hat{W}^\mathrm{on}_{t\rightarrow t+N}(\mathbf{w}_{0 \rightarrow t-1})):= \nonumber\\
    \min_{v_{t|t},\cdots,v_{t+N-1|t}}& \! \sum_{k=t}^{t+N-1}\!\! \ell(\bar{x}_{k|t}, u_{k|t}(\bar{x}_{k|t}))+Q_{t+N \rightarrow T}(\bar{x}_{t+N|t})  \nonumber \\
    \mathrm{s.t.,}~~~~~~&  x_{k+1|t}   = Ax_{k|t}+Bu_{k|t}(x_{k|t})+w_{k|t}, \\
    &\bar{x}_{k+1|t}  =  A\bar{x}_{k|t}+Bu_{k|t}(\bar{x}_{k|t}), \\
    &u_{k|t}(x_{k|t}) = Kx_{k|t}+v_{k|t},\label{eq:aff_pol_in_mpc2} \\
    &x_{t|t} = \bar{x}_{t|t} = x(t),~~ \\
    % &~~~~~~H_u u_{t|t} \leq h_u,  \\
    & H_u u_{k|t} \leq h_u,~\forall [w_{t|t},\dots, w_{k-1|t}] \in \mathbb{W}^{k-t}, \label{eq:input_constr} \\
    &[H]_i x_{k+1|t} \leq  h_i, ~\forall i\in\{1,\dots,p\}, \nonumber\\
    &\forall [w_{t|t},\dots,w_{k|t}] \in \hat{W}^{\mathrm{on},k+1}_{t \rightarrow k+1,i}, \label{eq:state_constr} \\
    &\forall k=\{t,\dots,t+N-1\}, \nonumber\\
     &\bar{x}_{t+N|t} \in \mathcal{X}_{F}^t,  \label{eq:terminal_set}
\end{align}
% \begin{align}
% % &V_t(x_t, \hat{W}^\mathrm{on}_{t\rightarrow t+N}(\mathbf{w}_{0 \rightarrow t-1})):= \nonumber\\
%     &\min_{u_{t|t},\cdots,u_{t+N-1|t}(\cdot)} \sum_{k=t}^{t+N-1} \ell(\bar{x}_{k|t}, u_{k|t})+Q_{t+N \rightarrow T}(\bar{x}_{t+N|t})  \nonumber \\
%     &~\mathrm{s.t.}~  x_{k+1|t}   = Ax_{k|t}+Bu_{k}(x_{k|t})+w_{k}, \\
%     &~~~~~~~\bar{x}_{k+1|t}  =  A\bar{x}_{k|t}+Bu_{k|t}(\bar{x}_{k|t}), \\
%     &~~~~~~~u_{k|t}(x_{k|t}) = Kx_{k|t}+v_{k|t},\label{eq:aff_pol_in_mpc2} \\
%     &~~~~~~~x_{t|t} = \bar{x}_{t|t} = x(t),~~ \\
%     &~~~~~~~H_u u_{k|t} \leq h_u, ~\forall w_{k} \in \mathbb{W}, \label{eq:input_constr} \\
%     &~~~~~~~H x_{k+1|t} \leq  h,~\forall [w_{t},\dots,w_{k}] \in \hat{W}^{\mathrm{on},k+1}_{t \rightarrow k+1}, \label{eq:state_constr} \\
%      &~~~~~~~x_{t+N|t} \in \mathcal{X}_F , ~\forall [w_{t},\dots,w_{t+N-1}] \in \hat{W}^{\mathrm{on},t+N}_{t \rightarrow t+N}, \label{eq:terminal_set}\\
%     &~~~~~~~ \forall k=\{t,t+1,\cdots, t+N-1\}, 
% \end{align}
\end{subequations}
where $\mathcal{X}_F^t$ and $Q_{t+N \rightarrow T}(\cdot)$ denote the terminal set for the nominal state and the terminal cost for the prediction step $t+N$, respectively. The state $x(t)$ is the realized closed-loop state. Notice our control input should satisfy hard input constraints robustly for all disturbances in $\mathbb{W}$ in \eqref{eq:input_constr}, while state constraints are satisfied robustly for only disturbances in $\hat{W}^{\mathrm{on},k+1}_{t \rightarrow k+1, i}$. 
After solving \eqref{eq:robust_reformulation_mpc}, we apply
\begin{align}\label{eq:mpc_cl}
    u_{t|t}^{\star} = Kx(t) + v^{\star}_{t|t}
\end{align}
to system \eqref{eq:sys} at time step $t$, with the first element $v^{\star}_{t|t}$ of the optimal solution sequence. 
% \subsection{Terminal Ingredients}
% We first define a robust positive invariant Set as follows:
% In addition, the safe policy $\pi(\cdot)$ is required in the form of \eqref{eq:aff_pol}. To find a Robust Control Invariant Set, we make use of a simple state feedback controller $u=Kx$. Under this controller, we can obtain the robust backward reachable set with respect to $\mathbb{W}$. The set is a Robust Control Invariant Set, which we can easily obtain. 
We pick $\mathcal{X}_F^t$ so that a terminal policy $Kx$ can make the terminal state evolve while satisfying the state constraint \eqref{eq:orig_co_scalar} for all future time steps $(\geq t+N+1)$. The construction of this set is elaborated in Section \ref{sec:real_dep_sets}.

\section{Offline Disturbance Sequence Sets}\label{sec:estimation}
In this section, we first illustrate properties of the sets $\hat{W}^{\mathrm{off},t}_{0\rightarrow t,i}$ for $t \in \{0,1,\dots, T-1\}, ~i \in \{1,\dots,p\}$ and then show how we construct these offline with collected data. 
% the required design specifications of sets $\hat{W}^{\mathrm{on}}_{t \rightarrow k+1}$, for all $k \in \{t,\cdots, t+N-1\}$ and then also demonstrate how we construct them using sampled data from many trials.
% In order to obtain $\hat{W}^{\mathrm{on}}_{t \rightarrow k+1}$, we first investigate $\hat{W}^{\mathrm{off}}_{0 \rightarrow k+1}$. 
These sets will be computed by propagating the uncertainty using the system dynamics and then computing the constraint tightening which satisfies the chance constraints~\eqref{eq:orig_co_scalar}. 
% by sampling method.

\subsection{Properties}\label{subsec:propagation}
% Note that choosing the sets ... such as
% \begin{align*}
% \mathbb{P}(w \in \hat{W}_k)\geq 1-\alpha,
% \end{align*}
% is sufficient condition to satisfy (4g). 
Consider being at $x_0$. Under the control policy \eqref{eq:aff_pol}, the chance constraints on the very next  state can be written as:
\begin{equation}\label{eq:chance_1st}
    \begin{aligned}
    &\mathbb{P}([H]_i x_{1}\leq h_i) \geq 1-\alpha, \\
    \iff & \mathbb{P}([H]_i (Ax_{0}+Bu_0+w_{0})\leq h_i) \geq 1-\alpha, \\
    \iff & \mathbb{P}([H]_i(A+BK)x_{0}+[H]_iBv_0-h_i \leq -[H]_iw_0)\\
    &\geq 1-\alpha, 
\end{aligned}
\end{equation}
for all $i \in \{1,2,\dots, p\}$. In \eqref{eq:chance_1st}, we can decouple the random variable terms and the decision variable terms.
Once we find a $\gamma_{0,i}$ satisfying: For $w_0 \in \mathbb{W}$ and $w_0 \sim f^{\mathrm{uk}}_{w}$,
% \begin{align*}
%     \mathbb{P}(w_0 \in \hat{W}_0)\geq 1-\alpha,
% \end{align*}
% which implies
\begin{align}\label{eq:W0_condition}
    \mathbb{P}([H]_iw_0 \leq \gamma_{0,i})\geq 1-\alpha,
\end{align}
then the set $\hat{W}^{\mathrm{off},1}_{0 \rightarrow 1,i}$ can be defined as:
\begin{align}
    \hat{W}^{\mathrm{off},1}_{0 \rightarrow 1,i} = \big \{w_0 ~\big|~ [H]_iw_0 \leq \gamma_{0,i},~ w_0 \in \mathbb{W} \big \}.
\end{align}
For the sake of brevity, we omit $w_t \in \mathbb{W}~\forall t$ in the disturbance sequence set definition in the rest of the paper.
If our control policy robustly satisfies $[H]_i(A+BK)x_{0}+[H]_iBv_0-h_i \leq -[H]_iw_0$, for all $w_0 \in \hat{W}^{\mathrm{off},1}_{0 \rightarrow 1,i}$, it also satisfies the chance constraint \eqref{eq:chance_1st} for $w_0 \in \mathbb{W}$ distributed according to $ f^{\mathrm{uk}}_{w}$, unknown probability distribution. Ideally we pick the smallest $\gamma_{0,i}$ tightly satisfying \eqref{eq:W0_condition}.
%  \begin{align}\label{eq:ordered_metric}
%  \gamma_{0,i} = \argmin \gamma ~~\textrm{s.t.}~~ \mathbb{P}(-[H]_iw_0 \leq \gamma)\geq 1-\alpha,
%  \end{align}
% where $w_0 \in \mathbb{W}$ is  distributed according to $ f^{\mathrm{uk}}_{w}$. 
Then we construct the disturbance sequence set to tightly satisfy the $i^{\mathrm{th}}$-chance constraint for time step $1$. 

% \noindent 
Using the same approach, 
% each inequality along the horizon can be separated into random variable terms which are combined with $A,B,K$ after expansion and decision variable terms. We examine these random variable terms based on collected disturbance samples in the later section.
% T
the $i^{\mathrm{th}}$-state constraint on a state $x_{t+1}$ formulated at time step $0$ can be expanded as:
\begin{equation}\label{eq:chance_kth}
\medmuskip=0mu
% \thinmuskip=0mu
% \thickmuskip=0mu
\begin{aligned}
    &\mathbb{P}([H]_i x_{t+1}\leq h_i) \geq 1-\alpha, \\
    \iff &\mathbb{P}([H]_i (Ax_{t}+Bu_t+w_{t})\leq h_i) \geq 1-\alpha, \\
    \iff &\mathbb{P}\Big([H]_i(A+BK)^tx_0+\dots+[H]_iBv_t-h_i \!\leq \!-([H]_iw_t+\\
    &[H]_i(A+BK)w_{t-1}+\dots+[H]_i(A+BK)^tw_0)\Big)\! \geq\! 1-\alpha. 
\end{aligned}
\end{equation}
From this expansion, we find the smallest $\gamma_{t,i}$ for each $t$ tightly satisfying:
\begin{align}\label{eq:varying_w_ineq_given0}
    \mathbb{P}\Big( [H]_i \big( \sum_{l=0}^{t}(A+BK)^{t-l} w_{l} \big) \leq \gamma_{t,i} \Big) \geq 1-\alpha.
\end{align}
% \begin{align}\label{eq:ordered_metric_k}
%  \gamma_t = &\argmin \langle \vec{\mathbf{1}},\gamma \rangle \nonumber \\
%  ~&\textrm{s.t.}~ \mathbb{P}(-H \big( \sum_{l=0}^{t}(A+BK)^{t-l} w_{l} \big) \leq \gamma)\geq 1-\alpha.
%  \end{align}
Then, the set $\hat{W}^{\mathrm{off},t+1}_{0 \rightarrow t+1,i}$ can be defined as:
\begin{align}\label{eq:offset_definition}
    &\hat{W}^{\mathrm{off},t+1}_{0 \rightarrow t+1,i} = \nonumber\\
    &\big \{[w_0,\dots,w_{t}]\big|[H]_i \big( \sum_{l=0}^{t}(A+BK)^{t-l} w_{l} \big) \leq \gamma_{t,i} \big \},
\end{align}
for each $i \in \{1,2,\dots, p\}$.

\subsection{Construction of $\gamma_{t,i}$ from Sampled Data}\label{subsec:construct_set}
% In many standard SMPC approaches, multivariate Chebyshev's inequality is used to obtain a conservative upper bound of disturbance propagation terms with specified probability. But, it requires the information of moments (e.g. mean, covariance) and it is not improved after collecting many samples. Moreover, since we do not have information of moments, we construct $\gamma_k$ from sampled data in a novel way. Briefly explaining, we construct a set which can include $1-\alpha$ portion of random variables, considering the criteria like \eqref{eq:varying_w_ineq_given0}. Then, we take an element-wise maximum of the criteria vector over the constructed set. It would be an estimation of $\gamma_k$.
% Firstly, we introduce the method to construct the estimated set including $1-\alpha$ portion of population from random variable samples. The set contains the variables which have the smaller value than $1-\alpha$ quantile in given ordered metric, with $1-\beta$ confidence.
% There are many sets $\hat{W}_0$ which satisfy \eqref{eq:W0_condition}. It depends how we decide $\alpha$ portion of $\mathbb{W} $ which is outside $\hat{W}_0$.
% More formally, 
In this section we explain how to obtain $\gamma_{t,i}$, defining  the sets $\hat{W}^{\mathrm{off},t+1}_{0 \rightarrow t+1, i}$, from sampled data.
% \blue{Basically, we construct a set which contains the propagated disturbances in $1-\alpha$ probability at least and obtain the upper bound as small as possible.} 
Firstly, we consider a method to construct a set containing random variables with probability at least $1-\alpha$, using sampled data.
Given a data set $D$ comprising samples of a random variable $d$,  we construct $X(d)$ containing $d$ variables sorted by a given metric \footnote{We use the $[H]_i \big( \sum_{l=0}^{t}(A+BK)^{t-l} w_{l} \big)$ of~\eqref{eq:varying_w_ineq_given0}
as the required ``metric for ordering". See \cite{hong2016approximating} for additional details.}, from $0^{\mathrm{th}}$-percentile to $100(1-\alpha)^{\mathrm{th}}$-percentile, with confidence $1-\beta$ for $0<\beta \ll 1$, satisfying:
\begin{align}\label{eq:set_general}
    \mathbb{P}_D(\mathbb{P}(d \in X(d)) \geq 1-\alpha))\geq 1-\beta,
\end{align}
% We obtain a set satisfying \eqref{eq:set_general} 
using the method in~\cite{hong2016approximating}.
$X(d)$ is chosen  as a polytope and is computed as the convex hull of a part of samples, which contains from $0^{\mathrm{th}}$-percentile to $100(1-\alpha)^{\mathrm{th}}$-percentile samples sorted by the given metric. Also, as the number of samples increases, $\beta$ will decrease.
% Now, we describe how to obtain the $\gamma_k$ utilizing this set construction m
% ethod. 
For time step $t=0$, when the samples of $w_0$ realizations and the metric $[H]_iw_0$ are given, $X_i(w_0)$ can be computed as explained in \eqref{eq:set_general}. $X_i(w_0)$ is computed differently for each $i^{\mathrm{th}}$ constraint and we have $p$ different sets.  Since $ \mathbb{P}([H]_iw_0 \leq \max\limits_{w_{[0]} \in X_i(w_0)} [H]_iw_{[0]}) \geq 1-\alpha$ with confidence $1-\beta$, we pick $\gamma_{0,i}$ as:
\begin{align}
    \gamma_{0,i} = \max_{w_{[0]} \in X_i(w_0)} [H]_iw_{[0]},~\forall i \in \{1,2,\dots, p\},
\end{align}
% It can be an estimation of $\gamma_0$ with $1-\beta$ confidence.
where $w_{[0]}$ denotes the decision variable for optimization as described in Section \ref{sec:notation}.
For any time step $t>0$, $\gamma_{t,i}$ is also obtained in the same approach as follows.
Define $y_t$ as:
\begin{align}\label{eq:def_yt}
    y_t = w_{t}+(A+BK)w_{t-1}+\cdots+(A+BK)^{t}w_0.
\end{align}
That is, $y_t$ is the random variable which describes the summation of propagated disturbance terms from $0$ to $t$. With the realized samples of $y_t$, we construct $X_i(y_{t})$ using the aforementioned method, and find $\gamma_{t,i}$ as the solution to:
\begin{align}\label{eq:yk_gammak}
    \gamma_{t,i} =& \max_{y_{[t]} \in X_i(y_t)} [H]_iy_{[t]} \nonumber\\
    =&\max_{\substack{[w_{[0]},\cdots,w_{[t]}]}} [H]_i \big( \sum_{l=0}^{t}(A+BK)^{t-l} w_{[l]} \big) \nonumber\\
    &~~~~~~\text{s.t.,}~~ y_{[t]}([w_{[0]},\cdots,w_{[t]}]) \in X_i(y_t)  
\end{align}
% Also we can store $y_k$ over previous trials because we need to calculate $y_k$ from $\mathbf{w}_{0 \rightarrow k}$ for \eqref{eq:varying_w_ineq} at each trial. It
Note $y_{t+1}$ can be computed as $y_{t+1} = w_{t+1}+(A+BK)y_t$.

\begin{remark}
% To keep high confidence enough, it requires sufficient data samples. We can collect the samples safely with conservative $\gamma_k$ considering all worst realizations, namely, equivalently to robust MPC. After many trials, we can switch to sample-based $\gamma_k$.
With a large number of samples and a long task horizon $T$, computation of $\gamma_{t,i}$ can become cumbersome. In that case, we can opt for an efficient way to compute a conservative $\gamma_{t,i}$ for large values of $t$ as follows. 

Given a $\gamma_{t,i}$, to compute the conservative $\gamma_{t+1,i}$ in an efficient manner, we use the following formula:
% Since the computation of $\gamma_k$ is implemented offline, the computation time issue is not much considered. But, with a large amount of samples and a long task horizon, it requires much time to compute. To address this issue, we always can utilize this conservative computation as:
\begin{align}\label{eq:conservative_way_gamma}
    \gamma_{t+1,i} := \max_{w \in \mathbb{W}} [H]_i(A+BK)^
    {t+1} w + \gamma_{t,i}. 
\end{align}
If $\gamma_{t,i}$ satisfies \eqref{eq:varying_w_ineq_given0}, the $\gamma_{t+1,i}$ obtained from \eqref{eq:conservative_way_gamma} also satisfies \eqref{eq:varying_w_ineq_given0} for $t+1$, since $\mathbb{P}(w_0 \in \{w_0|[H]_i(A+BK)^{t+1}w_0 \leq \max\limits_{w\in\mathbb{W}} [H]_i(A+BK)^{t+1}w \})=1$ and $\mathbb{P}([w_1,\dots,w_t] \in \{[w_1,\dots,w_t]|[H]_i \big( \sum_{l=1}^{t}(A+BK)^{t-l} w_{l} \big) \leq \gamma_{t,i} \})\geq 1-\alpha$. This $\gamma_{t+1,i}$ is enough to satisfy \eqref{eq:varying_w_ineq_given0} for $t+1$, although it is more conservative than the smallest $\gamma_{t+1,i}$ satisfying \eqref{eq:varying_w_ineq_given0}. 
Moreover, according to Assumption  \ref{assmp:strictly_stable}, we can find $\gamma_{t+1,i} \approx \gamma_{t,i}$ for sufficiently large $t$ because $(A+BK)^t w \approx 0$, so that we can set $\gamma_{t,i}$ as the fixed $\gamma_{\bar{t},i}$ for sufficiently large $t \geq \bar{t}$.
% After large enough $t$, we do not need tons of data to compute $\gamma_{t+1}$ so that we can reduce the computation time.
\end{remark}
% \noindent Then, we can get $\hat{W}_k (\subseteq \hat{W}_0)$ varying along the horizons for all $k$ by induction.
% \begin{remark}

Note that when the disturbance distribution is known, we can find conservative $\gamma_{t,i}$ analytically (e.g., utilization of Chebyshev inequality in \cite{cannon2012stochastic}, etc.). This guarantees constraint satisfaction in closed-loop with confidence $1$. We approximately obtain $\gamma_{t,i}$ with our sampling method at the cost of constraint satisfaction with a confidence level less than $1$, so that we can tackle the unknown distribution of disturbance. With sufficient data, the confidence level will be $\sim1$. 
% and the obtained $\gamma_t$ may be less conservative than using the conservative bound.
% \end{remark}

% because $\hat{W}_k = \hat{W}_0$ is a sufficient condition for \eqref{eq:estimate_set_constr}, given proper $\hat{W}_{k-1},\cdots,\hat{W}_0$. For notational convenience, we define the sequence of estimated sets as:
% Once we obtain the corresponding set, 
% % we solve for element-wise maximum value over the set we found. We reformulate the chance constraint to deterministic constraint with solved maximum values.
% we take the robust MPC reformulation with the proposed estimated sets. 

% Basically, we need to obtain $\hat{W}_0, \cdots, \hat{W}_{T-1}$ satisfying each constraint corresponding to the horizon (i.e. \eqref{eq:estimate_set_constr} at $k+1$) for entire task horizon. They can be computed offline after finishing one iteration. In some conditions, we can find a proper $\hat{W}_s$ which can be substituted with $\hat{W}_k ~~\forall k > \bar{k}$, given sufficiently large $\bar{k}$. After substituting, all the constraints after $\bar{k}$ are still satisfied.  

% \textcolor{blue}{Remove 3.2. Say that how to construct these W sets is shown in the Appendix and for that we use the method of xyz.... Go straight to Section 4 here, which is more important.}

% \noindent We show how to construct the set from the collected data and how to find the $\hat{W}_s$ in the 

\section{Realization Dependent Online Sets}\label{sec:real_dep_sets}
Consider MPC problem \eqref{eq:robust_reformulation_mpc} at current time step $t$ with MPC prediction step $k \in \{t,t+1,\dots,t+N-1\}$. We explain how to construct $\hat{W}^{\mathrm{on},k+1}_{t \rightarrow k+1,i}$, $\forall i \in \{1,\dots,p\}$ and $\mathcal{X}^t_{F}$ online. We obtain $\hat{W}^{\mathrm{on},k+1}_{t \rightarrow k+1,i}$ by utilizing past disturbance realizations and $\gamma_{k,i}$ which defines $\hat{W}^{\mathrm{off},k+1}_{0 \rightarrow k+1,i}$, in order to satisfy the state constraints robustly for all disturbance sequences in $\hat{W}^{\mathrm{off},k+1}_{0 \rightarrow k+1,i}$ in closed-loop. The terminal set for the nominal state, $\mathcal{X}^t_{F}$,  is constructed specifically to ensure recursive feasibility of \eqref{eq:robust_reformulation_mpc}.

\subsection{Construction of $ \hat{W}^{\mathrm{on},k+1}_{t \rightarrow k+1,i}$ for $k \in \{t,\dots,t+N-1\}$}\label{sec:Won_description}
After obtaining the $\gamma_{k,i} $ using \eqref{eq:yk_gammak} and the past realized disturbances  $\mathbf{w}_{0\rightarrow t-1}=[w(0),\dots, w(t-1)]$ at time step $t$, we can construct $\hat{W}^{\mathrm{on},k+1}_{t \rightarrow k+1,i}$ using $\gamma_{k,i}$ and $\mathbf{w}_{0\rightarrow t-1}$. $\hat{W}^{\mathrm{on},k+1}_{t \rightarrow k+1,i}$ is constructed as a set of $[w_{t|t},\dots, w_{k|t}]$ which satisfy the inequalities discussed next. For brief descriptions, we introduce the simplified terms for $k \in \{t,\dots,t+N-1\}$. Recall the notations from Section \ref{sec:notation}.
\begin{align*}
    &a_{t,i}^{k} = [H]_i \big( \sum_{l=t}^{k}(A+BK)^{k-l} w_{l|t} \big ),\\[-1mm]
    &M_{t,i}^{k} = \max\limits_{w_{[l]} \in \mathbb{W} ~\forall l} [H]_i \big( \sum_{l=t}^{k}(A+BK)^{k-l} w_{[l]} \big ),\\[-1mm]
    &m_{t,i}^{k} = \min\limits_{w_{[l]} \in \mathbb{W} ~\forall l} [H]_i \big( \sum_{l=t}^{k}(A+BK)^{k-l} w_{[l]} \big ),\\[-1mm]
    &b_{0,i}^{t-1,k}=[H]_i \big (\sum_{l=0}^{t-1} (A+BK)^{k-l} w(l) \big),
\end{align*}
where $a_{t,i}^{k}$ denotes description of accumulated disturbances from $t$ to $k$, $M_{t,i}^{k}$ denotes the maximum of admissible accumulated disturbances from $t$ to $k$,  $m_{t,i}^{k}$ denotes the minimum of admissible accumulated disturbances from $t$ to $k$, and $b_{0,i}^{t-1,k}$ denotes the accumulation of realized disturbances from $0$ to $t-1$ for $k$ prediction step.
Then we define
% \begin{align}\label{eq:varying_w_ineq}
%     H \big( \sum_{l=t}^{k}(A+BK)^{k-l} w_{l} \big ) \leq \gamma_k - H \Big (\sum_{l=0}^{t-1} (A+BK)^{k-l} w(l) \Big).
% \end{align}
\begin{align}\label{eq:varying_w_ineq}
&\hat{W}^{\mathrm{on},k+1}_{t \rightarrow k+1,i} \nonumber \\ 
&= \Bigg\{ [w_{t|t},\dots,w_{k|t}] \Bigg|
\begin{cases}
    &\!\!\!\!a_{t,i}^{k} \leq M_{t,i}^{k},~~~~~~\text{if condition~(i)} \\
    &\!\!\!\!\varnothing,~~~~~~~~~~~~~~~~~\text{if condition~(ii)} \\
    &\!\!\!\! a_{t,i}^{k} \leq \gamma_{k,i} \!-\! b_{0,i}^{t-1,k},~\text{otherwise}
\end{cases}
\Bigg\}
\end{align}
% \begin{align}\label{eq:varying_w_ineq}
% \begin{cases}
%     &H \big( \sum_{l=t}^{k}(A+BK)^{k-l} w_{l} \big ) \\
%     &\leq \max\limits_{w_{[l]} \in \mathbb{W} ~\forall l} H \big( \sum_{l=t}^{k}(A+BK)^{k-l} w_{[l]} \big ),\text{if condition~(i)} \\
%     &\varnothing,~~~~~~~~~~~~~~~~~~~~~~~~~~~~~~~~~~~~~~~~~~~~~\text{if condition~(ii)} \\
%     & H \big( \sum_{l=t}^{k}(A+BK)^{k-l} w_{l} \big ) \\
%     & \leq \gamma_k - H \Big (\sum_{l=0}^{t-1} (A+BK)^{k-l} w(l) \Big),~~~\text{otherwise}
% \end{cases}
% \end{align}
% \begin{align}\label{eq:varying_w_ineq}
%     \begin{cases}
%     H \big( \sum_{l=t}^{k}(A+BK)^{k-l} w_{l} \big ) \\
%     \leq \max\limits_{w_{[l]} \in \mathbb{W} ~\forall l} H \big( \sum_{l=t}^{k}(A+BK)^{k-l} w_{[l]} \big ) ~ \text{if condition~(i)}
%      \\
%     \varnothing ~ \text{if condition~(ii)},\\
%     H \big( \sum_{l=t}^{k}(A+BK)^{k-l} w_{l} \big ) \\
%     \leq   \gamma_k - H \Big (\sum_{l=0}^{t-1} (A+BK)^{k-l} w(l) \Big)  \text{if condition~(iii)}\\
%     \end{cases}
% \end{align}
for $i \in \{1,2,\dots, p\}$, where the conditions (i) and (ii) are:
\begin{enumerate}[(i)]
    \item  $b_{0,i}^{t-1,k} \leq \gamma_{k,i} - M_{t,i}^{k} $, 
    \item  ${b_{0,i}^{t-1,k}} > \gamma_{k,i}- m_{t,i}^{k}$.
\end{enumerate}
% \begin{enumerate}[(i)]
%     \item     $\gamma_k \geq H \Big (\sum_{l=0}^{t-1} (A+BK)^{k-l} w(l) \Big) + \max\limits_{w_{[l]} \in \mathbb{W} ~\forall l} H \big( \sum_{l=t}^{k}(A+BK)^{k-l} w_{[l]} \big )$
%     \item $\exists i$, $(\gamma_k)_{(i)} \leq \bigg(H \Big (\sum_{l=0}^{t-1} (A+BK)^{k-l} w(l) \Big) +\min\limits_{w_{[l]} \in \mathbb{W} ~\forall l} H \big( \sum_{l=t}^{k}(A+BK)^{k-l} w_{[l]} \big ) \bigg)_{(i)}$.
% \end{enumerate}

\noindent 
The intuitive explanation behind the set description \eqref{eq:varying_w_ineq} is presented next. The MPC controller is designed to satisfy the state constraints for all admissible disturbance sequences, which belong to $\hat{W}^{\mathrm{off},k+1}_{0 \rightarrow k+1,i}$, with the past disturbance realizations up to time $t-1$.  
\begin{itemize}
    \item If  the weighted sum of the disturbance sequence is small, resulting in condition (i) for $k$, then the controller tries to be robust against the entire $\mathbb{W}$ for future $k-t+1$ steps since all admissible disturbance sequences can belong to the offline computed set.
% problem \eqref{eq:robust_reformulation_mpc} must be robust against the entire portion of $\mathbb{W}$ for future $k-t+1$ steps. This ensures satisfaction of \eqref{eq:orig_co}. 
So the set of disturbance sequences the controller would be robust against, will be decided as the largest set, equivalent to the first set description ($=\{[w_{t|t},\dots,w_{k|t}] \big| a_{t,i}^{k} \leq M_{t,i}^{k}\}$) in \eqref{eq:varying_w_ineq}.
\item On the other hand, if the weighted sum of the disturbance sequence is very large, resulting in condition (ii) for $k$, then the accumulation of the disturbances from $0$ to $k$ cannot be less than or equal to $\gamma_k$ regardless of future disturbances from $t$ to $k$. In this case, the controller does not need to satisfy the state constraint at prediction step $k+1$ for any disturbance sequences. The user-specified upper bound $\alpha$ allows violation of constraints for these extreme disturbance realizations. These extreme disturbance realizations are already outside the offline disturbance sequence set defined by $\gamma_{k,i}$ in \eqref{eq:offset_definition}. It is equivalent to the second set description (i.e., empty) in \eqref{eq:varying_w_ineq}.

% If the accumulated disturbances are small, resulting in condition (i) for $k$, then the controller tries to be robust against the entire $\mathbb{W}$ for future $k-t+1$ steps since all disturbance sequences belong to the offline set.
% % problem \eqref{eq:robust_reformulation_mpc} must be robust against the entire portion of $\mathbb{W}$ for future $k-t+1$ steps. This ensures satisfaction of \eqref{eq:orig_co}. 
% So the set of disturbance sequences the controller would be robust against, will be decided as the largest set, leading to the first set description ($=\{[w_t,\dots,w_k] \big| a_{t}^{k} \leq M_{t}^{k}\}$) in \eqref{eq:varying_w_ineq}.
% On the other hand, if accumulated disturbances are too large, resulting in condition (ii) for $k$, \blue{then the accumulation of the disturbances up to $k$ cannot be less than or equal to $\gamma_k$ regardless of future disturbances from $t$ to $k$. In this case, the controller does not need to satisfy the state constraint at prediction step $k+1$ for any disturbance sequences. The user-specified upper bound $\alpha$ allows violation of constraints for these extreme disturbance realizations. These extreme disturbance realizations are already outside the offline disturbance sequence set defined by $\gamma_k$ in \eqref{eq:offset_definition}. It leads to the second set description (i.e., empty) in \eqref{eq:varying_w_ineq}.}
% can afford to be not robust to anything in the remainder of the trajectory, and yet satisfy \eqref{eq:orig_co}. 
\item If the weighted sum of the disturbance sequence is not extreme so that neither condition (i) nor condition (ii) holds, we can adjust the set of disturbance sequences depending on previous realizations as the third set description inequality ($=\{[w_{t|t},\dots,w_{k|t}] \big| a_{t,i}^{k} \leq \gamma_{k,i} - b_{0,i}^{t-1,k}\}$) in \eqref{eq:varying_w_ineq}.
\end{itemize}
Thus, the adjusted $\hat{W}^{\mathrm{on},k+1}_{t \rightarrow k+1,i}$ in \eqref{eq:state_constr} can be computed online by \eqref{eq:varying_w_ineq}.
The satisfaction of \eqref{eq:orig_co_scalar} is ensured with $1-\beta$ confidence due to the properties of chosen $\gamma$ parameters in Section~\ref{subsec:propagation}. 

\subsection{Construction of the terminal set $\mathcal{X}_F^{t}$}\label{rmk:terminal}
To keep recursive feasibility of \eqref{eq:robust_reformulation_mpc}, we define a terminal set $\mathcal{X}_F^t$ for prediction step $t+N$ at current time step $t$ as follows. If the nominal state enters this set at time $t+N$, the state chance constraints are satisfied at all future time steps for the states which evolve through the closed-loop system with a policy $u=Kx$.
Define the terminal set $\mathcal{X}_F^0$ at time step $0$ first as:
\begin{equation}\label{eq:initil_terminal_set}
    \begin{aligned}
    \mathcal{X}_{F}^{0} &=\Big\{x ~\Big|~ [H]_i A_{\mathrm{cl}}^{l+1}x \leq h_i - \gamma_{N+l,i}, \\
    & ~~~~~~\forall l = \{0,\dots,T-N-1\}, 
    ~\forall i = \{1,\dots,p\} \Big\} 
    % &= \big\{x ~\big|~ Lx \leq g \big\}
    \end{aligned}
\end{equation}
with $A_\mathrm{cl} = (A+BK)$.
By recalling \eqref{eq:chance_kth}, we find that the nominal state $\bar{x}_{N}=\bar{x}_{N|0}$ in $\mathcal{X}_F^0$ evolves through the closed-loop system while satisfying all state chance constraints for all future time steps. The set $\mathcal{X}_F^0$ corresponds to the terminal set which is used in \cite{cannon2012stochastic}.

\noindent Considering $\bar{x}_{t+N|0} = \bar{x}_{t+N|t}-A_{\mathrm{cl}}^N \sum_{j=0}^{t-1} A_{\mathrm{cl}}^{t-1-j} w(j)$, we can construct the terminal set $\mathcal{X}_F^{t}$ for $\bar{x}_{t+N|t}$ from the set $\mathcal{X}_F^{0}$ as:
\begin{equation}\label{eq:terminal_set_def}
    \begin{aligned}
    \mathcal{X}_{F}^{t} &\!=\!\Big\{x \Big| [H]_i A_{\mathrm{cl}}^{l+1}x \!\leq\! h_i \!-\! \gamma_{N\!+t\!+l,i}\!+\![H]_i\!  \sum_{j=0}^{t-1} A_{\mathrm{cl}}^{N\!+t\!+l\!-j} w(j), \\
    & ~~~~~~\forall l = \{0,\dots,T-N-1-t\}, 
    ~\forall i = \{1,\dots,p\} \Big\}
    \end{aligned}
\end{equation}
% \begin{align}\label{eq:terminal_set_def}
%     \mathcal{X}_F^{t} = \big\{x ~\big|~ L(A_{\mathrm{cl}}^{-1})^t x\leq g+L(A_{\mathrm{cl}}^{-1})^t \sum_{j=0}^{t-1} A_{\mathrm{cl}}^{N+t-1-j} w(j) \big\}.
% \end{align}
In practice, the set can be computed by only $\hat{N}$ inequalities ( $l=0,\dots,\hat{N}-1$), provided $\hat{N}$ is sufficiently large, according to \cite{gilbert1991linear}.
We will show recursive feasibility of \eqref{eq:robust_reformulation_mpc} in the next section's Theorem \ref{thm:feasible} with the terminal set $\mathcal{X}_F^{t}$ in \eqref{eq:terminal_set_def}.
% Since $\mathcal{X}_F$ is robust invariant under $u=Kx$, the terminal constraint satisfaction holds persistently.
%  That is, using the simplified terms in Sec. \ref{sec:Won_description}, $\hat{W}^{\mathrm{on},t+N}_{t \rightarrow t+N} = \{[w_t,\dots,w_{t+N-1}] \big| a_{t}^{t+N-1} \leq M_{t}^{t+N-1}\}$. 
% Therefore, we use $\gamma^{\mathrm{temp}}_{t+N}$ as: 
% \begin{align*}
%     \gamma_{t+N}^{\mathrm{temp}}:=\max\limits_{y \text{~s.t.~} Hy \leq \gamma_{t+N-1}}  H(A+BK)y+\max\limits_{w \in \mathbb{W}} Hw,
% \end{align*}
% instead of $\gamma_{t+N}$ when we construct the $\hat{W}^{\mathrm{on},t+N}_{t \rightarrow t+N}$ in \eqref{eq:terminal_set}. 
% It is used to prove recursive feasibility in Theorem \ref{thm:feasible} in the next section.

% Once we achieve $\gamma_k$, $\hat{W}_{t \rightarrow k+1}$ is decided by $\gamma_k$ and realizations $\mathbf{w}_{0 \rightarrow t-1}$ in \eqref{eq:varying_w_ineq}.
% Due to the hard input constraint robust to all disturbances in \eqref{eq:input_constr},
% we can always apply $u_t(x_t) = Kx_t+v_t$, irrespective of $x_t$ in closed-loop.
% The construction of vectors $\gamma_k$ and these sets $\hat{W}_{t \rightarrow k+1}(\mathbf{w}_{0 \rightarrow t-1})$ from sampled data is described in detail next.

\section{Algorithm and Properties}\label{sec:feas_properties}
Our algorithm is summarized in Algorithm \ref{alg:1}.
\setlength{\textfloatsep}{0pt}
 \begin{algorithm}[t]
 \caption{SMPC with Realization-Adaptive Constraint Tightening}
 \begin{algorithmic}[1]
 \renewcommand{\algorithmicrequire}{\textbf{Input:}}
 \renewcommand{\algorithmicensure}{\textbf{Output:}}
 \REQUIRE Sampled data $y_t~(t=0,1,\dots,T-1)$ in \eqref{eq:def_yt}
 \ENSURE  Control Policy
 \\ \textit{Offline} :
  \STATE     Compute $\gamma_{t,i}$\!\! $~\forall t \in \{0,1,\dots,T-1\}, \forall i \in \{1,\dots,p\}$ which defines $\hat{W}^{\mathrm{off},t+1}_{0 \rightarrow t+1,i}$ of \eqref{eq:offset_definition}, using \eqref{eq:yk_gammak} of  Sec-
  
  tion \ref{subsec:construct_set}.
 \\ \textit{Online} :
  \FOR {$t = 0$ to $T-N$}
  \STATE Construct $\hat{W}^{\mathrm{on},k+1}_{t \rightarrow k+1,i}$ from $\gamma_{k,i}$, $\forall k \in \{t,\dots,t+N-1\}, \forall i \in \{1,\dots ,p\}$ and previous realizations $\mathbf{w}_{0 \rightarrow t-1}$ using \eqref{eq:varying_w_ineq} 
    \STATE Solve MPC problem \eqref{eq:robust_reformulation_mpc} at time $t$ with $\hat{W}^{\mathrm{on},k+1}_{t \rightarrow k+1,i}$
    \STATE Apply the first optimal policy $u^\star_{t|t}(x(t))$ to \eqref{eq:sys}
    % \STATE Store the realization of $y_t$ computed from \eqref{eq:def_yt} with $\mathbf{w}_{0 \rightarrow t}$
  \ENDFOR
%  \RETURN 
\STATE
For $t =\! T-\!N+1$ to $T-\!1$, apply the corresponding control policy from the optimal policy sequence $[u^\star_{T-\!N+1|T-\!N},\dots,u^\star_{T-\!1|T-\!N}]$ solved at time $T-\!N$.
 \end{algorithmic}\label{alg:1}
 \end{algorithm}
% \begin{enumerate}[1)]
%     \item Initialize $\gamma_k$ with evaluation of the worst cases  for all $k = (0,1,\dots,T)$ from (),()..
%     \item (Online MPC) Every time we solve MPC problem \eqref{eq:robust_reformulation_mpc} at time $t$, construct $\hat{W}^\mathrm{on}_{t \rightarrow k+1}$ from $\gamma_k$ ($k =(t,\dots,t+N)$ and previous realizations $\mathbf{w}_{0 \rightarrow t-1}$ using \eqref{eq:varying_w_ineq}
%     \item Until enough data is collected, repeat 2)
%     \item (Offline) Once the enough trials are finished, reload disturbance samples in the form of $y_k ~(k=0,1,\dots,T)$ 
%     \item Obtain $X(y_k)$ and compute $\gamma_k$ for all $k$ 
%     \item Goto 2) again
% \end{enumerate}
So far, we obtain the disturbance sequence sets for robust MPC policy in order to satisfy the chance constraints. Since the probability distribution of disturbance is unknown, we cannot obtain the smallest  $\gamma_{t,i} ~\forall t \in \{0,\dots,T-1\}, ~\forall i \in \{1,\dots,p\}$ satisfying \eqref{eq:varying_w_ineq_given0} in analytical way. Practically, we should use the sampling method to obtain the approximate $\gamma_{k,i}$ in \eqref{eq:yk_gammak} with $1-\beta$ confidence, in order to construct the disturbance sequence sets.
It implies that the chance constraint satisfaction in closed-loop can hold with $1-\beta$ probability. 
However, if we collect sufficiently large amount of samples, $\beta$ can be made negligibly small. Therefore, we consider the ideal case where $\beta$ can be regarded as zero and the smallest $\gamma_{t,i} ~\forall t \in \{0,\dots,T-1\}, ~\forall i \in \{1,\dots,p\}$ satisfying \eqref{eq:varying_w_ineq_given0} can be obtained. Next  we show properties of Algorithm~\ref{alg:1} when we assume $\beta = 0$. With these disturbance sequence sets, we show that, for the state obtained from the system \eqref{eq:sys} in closed-loop with our MPC policy \eqref{eq:mpc_cl}, there exists a MPC policy which satisfies \eqref{eq:state_constr}, \eqref{eq:terminal_set} recursively at every time step, while satisfying \eqref{eq:orig_co_scalar}.
\begin{theorem}[Recursive Feasibility]\label{thm:feasible}
Let \eqref{eq:robust_reformulation_mpc} be feasible at time $t=0$, with the disturbance sequence sets $\hat{W}^{\mathrm{on},k+1}_{0 \rightarrow k+1,i} ~(k = 0,\dots,N-1, ~i = 1,\dots,p)$ computed using Algorithm \ref{alg:1}. Then, \eqref{eq:robust_reformulation_mpc} remains feasible at all time steps $t\geq 0$ robustly for any admissible $x(t)$, if the state $x(t)$ is obtained by applying the closed-loop MPC control \eqref{eq:mpc_cl} to system \eqref{eq:sys}.
 \end{theorem}
\begin{proof} See Appendix.
\end{proof}

% \begin{theorem}[Larger ROA]\label{thm:roa}
% We define the Region of Attraction (ROA) as the set of all initial states $x_0$ where the MPC problem \eqref{eq:robust_reformulation_mpc} is feasible.
% % Assume the smallest $\gamma_{t,i} ~\forall t \in \{0,\dots,T-1\}, ~\forall i \in \{1,\dots,p\}$ can be obtained with 1 confidence (i.e., $\beta=0$). 
% Then, the ROA of the proposed approach is a superset of the ROA of the existing approach in \cite{cannon2012stochastic}.
% \end{theorem}
% \begin{proof} See Appendix.
% \end{proof}
% \begin{theorem}[Constraint Satisfaction in closed-loop]
% Consider the system \eqref{eq:sys} under the closed-loop MPC control \eqref{eq:mpc_cl} computed using Algorithm. \ref{alg:1}. The resulting state $x(t)$ in closed-loop satisfies the constraints \eqref{eq:orig_co}.
% \end{theorem} 
% \begin{proof} See Appendix.
% \end{proof}
% We have successfully proposed the algorithm that can satisfy the original chance constraint with bounded support by building disturbance sequence sets from sampled data. In the next section, we apply this proposed MPC to some simulation examples and compare it with the existing approach.
% However, this iteration data can further be used to enlarge $\mathcal{X}_F$. This is the focus of our next section. 

\section{Numerical Simulations}
In this section, we numerically compare the performance of the proposed SMPC in Algorithm~\ref{alg:1} with the existing feasible SMPC approach in \cite{cannon2012stochastic}.  
% We apply Algorithm~\ref{alg:1} and \cite{cannon2012stochastic} to several examples. 
% For this, we consider three different sets of system dynamics and constraint matrices. If examples have few constraints or no input constraints, we add more constraints to each example so that we can compare the area of Region of Attraction (ROA).
% In order to show comparison clearly, we also apply the existing feasible SMPC approach  \cite{cannon2012stochastic} to the same example.
% Since we need the distribution information of the disturbance to implement SMPC of \cite{cannon2012stochastic}, we use the distribution information directly for \cite{cannon2012stochastic} although we use the sampling method for our proposed approach due to unknown distribution.
% Although we assume unknown distribution information of the disturbance in this paper, we use the distribution information to compute the constraint tightening using Chebyshev’s inequality when we implement SMPC of \cite{cannon2012stochastic} .
% Also, we try various disturbance distributions within the same support to check if our sampled method works properly.
With Algorithm~\ref{alg:1}, we find MPC solutions to the optimal control problem \eqref{eq:original_opti}
% \begin{subequations}\label{eq:example_ocp}
% \begin{align}
%     ~\min_{u_{0},\cdots,u_{T-1}(\cdot)}&~~ \sum_{k=0}^{T-1} \ell(\bar{x}_{k}, u_{k}(\bar{x}_{k}))+\ell_F(\bar{x}_T)  \nonumber \\
%     \mathrm{s.t.,}&  ~~~x_{k+1}  = Ax_{k}+Bu_{k}(x_{k})+w_{k}, \\
%     & ~~~\bar{x}_{k+1}  =  A\bar{x}_{k}+Bu_{k}(\bar{x}_{k}), \\
%     & ~~~u_{k}(x_{k}) = Kx_{k}+v_{k},\\
%     & ~~~ H_u u_{k}\leq h_u,   \\
%     & ~~~\mathbb{P}(H x_{k+1} \leq  h_x ) \geq 1-\alpha,  \\
%     & ~~~x_{0} = \bar{x}_{0} = x_S,~~ \\
%     & ~~~\forall k=\{0,1,\cdots, T-1\}, 
% \end{align}
% \end{subequations}
with $T=15, ~N=6$ and we choose three sets of $A,B,H,h_x,H_u,h_u$, denoted with (E1), (E2), (E3), respectively in the Appendix. The cost function is quadratic, as shown in the Appendix, with penalty matrices $Q=Q_F=I_2$, $R=1$ and $x_{\mathrm{ref}}=[0,0]^\top$. We set the terminal cost of MPC as the LQR cost-to-go.  
% \begin{enumerate}[(a)]
%     \item \label{param:1st} iterative SMPC \cite{rosolia2018stochastic}
%     \begin{align*}
%         A = \begin{bmatrix} 1.2 & 1.5\\0 & 1.3 \end{bmatrix}, B = \begin{bmatrix} 0\\1 \end{bmatrix}, H_x = \begin{bmatrix} 1 & 0\\-1 & 0\\0 &1\\0 &-1 \end{bmatrix}, \\
%         h_x = \begin{bmatrix} 20\\20\\20\\20 \end{bmatrix},
%         H_u = \begin{bmatrix} 1\\-1 \end{bmatrix},
%          h_u = \begin{bmatrix} 6\\6 \end{bmatrix}
%     \end{align*}
%     \item \label{param:2nd} Strongly Feasible SMPC \cite{korda2011strongly}
%     \begin{align*}
%         A = \begin{bmatrix} 1 & 0\\1 & 1 \end{bmatrix}, B = \begin{bmatrix} 1\\2 \end{bmatrix}, H_x = \begin{bmatrix} 1 & 0\\-1 & 0\\0 &1\\0 &-1 \end{bmatrix}, \\
%         h_x = \begin{bmatrix} 20\\20\\20\\20 \end{bmatrix},
%         H_u = \begin{bmatrix} 1\\-1 \end{bmatrix},
%          h_u = \begin{bmatrix} 2\\2 \end{bmatrix}
%     \end{align*}
%     \item \label{param:3rd} Stochastic Tube \cite{cannon2012stochastic} 
%     \begin{align*}
%         A = \begin{bmatrix} 1.6 & 1.1\\-0.7 & 1.2 \end{bmatrix}, B = \begin{bmatrix} 1\\1 \end{bmatrix}, H_x = \begin{bmatrix} 1 & 0\\-1 & 0\\0 &1\\0 &-1 \end{bmatrix}, \\
%         h_x = \begin{bmatrix} 20\\20\\20\\20 \end{bmatrix},
%         H_u = \begin{bmatrix} 1\\-1 \end{bmatrix},
%          h_u = \begin{bmatrix} 6\\6 \end{bmatrix}
%     \end{align*}
% \end{enumerate}
% \{w\in \mathbb{R}^2| -1 \leq w_{[i]} \leq 1 ~\forall i \} =
We used the remaining parameters as: $\mathbb{W}= [-1,1]\times[-1,1],~ \alpha = 0.1,~  K=K_{\mathrm{lqr}}$ (with the chosen $Q, R$). For the sampling method, we use $500$ disturbance samples.
% and the state cost is a quadratic cost for the reference $x_{\mathrm{ref}}$ and the input $u_k$ with $Q = 2I$ and $R = 10$. % We utilize a simple state feedback $u=Kx$ to obtain the robust positively invariant set under the policy as the terminal set with state/input constraints.
The optimization problems are formulated with YALMIP interface \cite{lofberg2004yalmip} in MATLAB, and we use Gurobi \cite{optimization2018gurobi} to solve the associated quadratic programs for control synthesis.

\subsection{Comparison of Approximate ROA}
% Moreover, as the constraint on state and input gets tighter, our proposed SMPC gives a larger Region of Attraction. 
We compare the area of approximate ROA between two approaches for three examples respectively. To compute the approximate ROA, we choose the method in \cite{bujarbaruah2021simple}. The approximate ROAs of the proposed MPC are about 7\% - 35\% larger in volume than that of the SMPC in \cite{cannon2012stochastic}. They are described in the Fig. \ref{fig:comparison_ROA_1}-\ref{fig:comparison_ROA_3}. 
\begin{figure}[h]
    \centering
    \includegraphics[width=0.8\columnwidth]{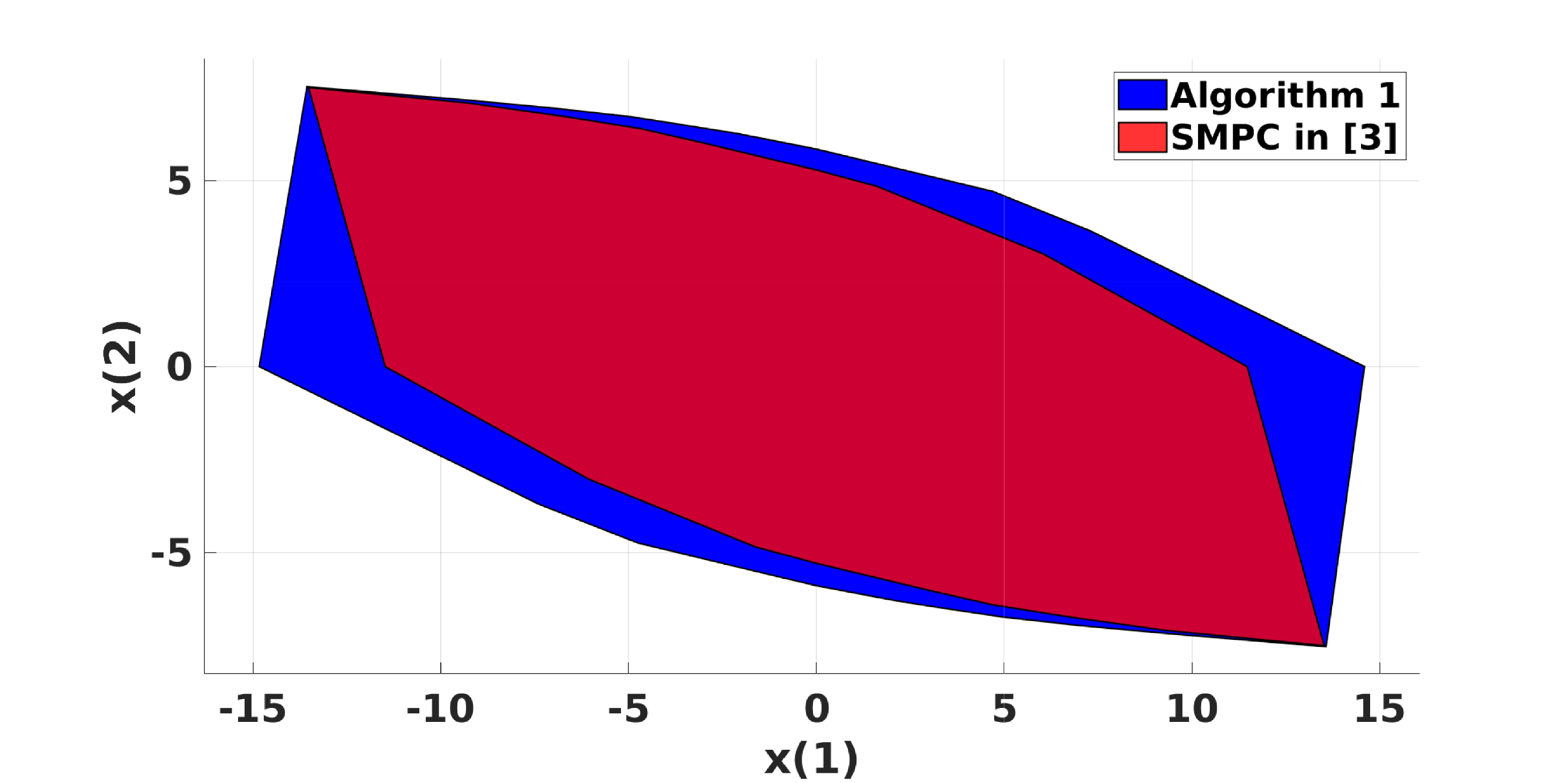}
    \caption{Comparison of approximate ROA with (E\ref{param:1st})}
    \label{fig:comparison_ROA_1}
\end{figure}
\begin{figure}[h]
    \centering
    \includegraphics[width=0.8\columnwidth]{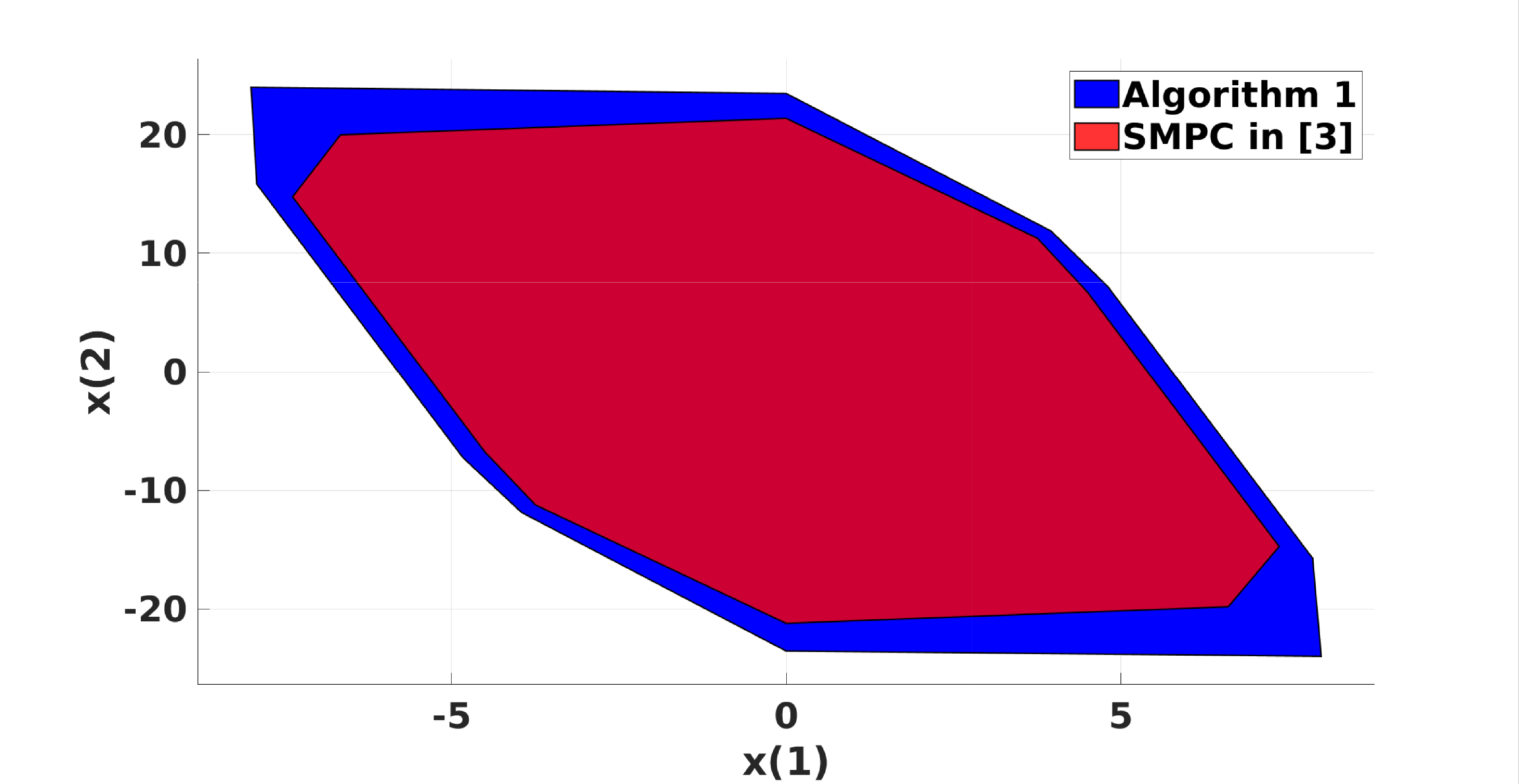}
    \caption{Comparison of approximate ROA with (E\ref{param:2nd})}
    \label{fig:comparison_ROA_2}
\end{figure}
\begin{figure}[h!]
    \centering
    \includegraphics[width=0.8\columnwidth]{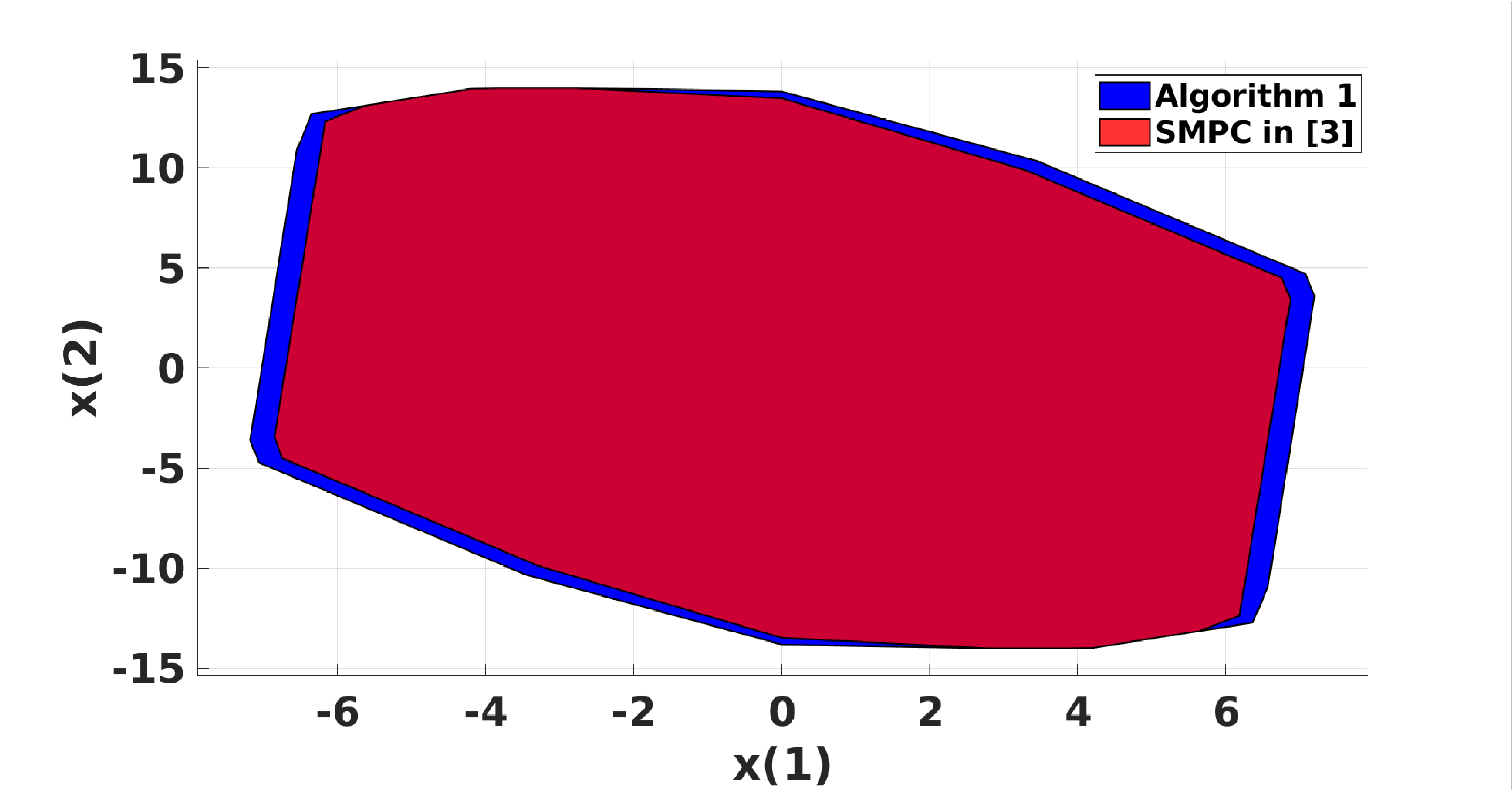}
    \caption{Comparison of approximate ROA with (E\ref{param:3rd})}
    \label{fig:comparison_ROA_3}
\end{figure}
% It can vary depending on $\alpha$ probability and system dynamics (i.e., $A,B$). 
\subsection{Comparison of Performance}

In this section, we focus on xmaple E2 and compare the average closed-loop cost of a realized trajectory  over multiple  realizations of the disturbance
when the proposed approach and the SMPC approach in \cite{cannon2012stochastic} is used 
% Moreover, the trajectories from the proposed approach will be compared with the trajectories of the feasible SMPC approach. 
% We try multiple sets of Monte-Carlo simulations so that we can see our proposed approach would be better, even in different sets of samples. 
We perform 10 trials and each trial contains different offline samples. In each trail, we get an average over 100 Monte-Carlo draws of the realized trajectory from the $x_S = [-5,19]^\top$.
For comparison of performance when initial states change, several initial states $x_S$ are sampled in the ROA. Our proposed approach provides about $1\sim6 \%$ lower average closed-loop costs than \cite{cannon2012stochastic}, as shown in Table. \ref{table:cost} 
\begin{table}[h!]
\centering
    \begin{tabular}{|c|| c c |c||} 
 \hline
  & Algorithm 1 & SMPC of \cite{cannon2012stochastic} & Improvement(\%) \\ [0.5ex] 
 \hline
 Avg. cost & 744.53 & 779.26 & 4.5 \\  [0.5ex] 
 \hline
 Best. cost & 730.67 & 772.18 & 5.7 \\ [0.5ex] 
 \hline
\end{tabular}
\caption{Closed-loop costs comparison.}
\label{table:cost}
\end{table}

In terms of computation time, our proposed approach requires almost the same online MPC computation time with the SMPC in \cite{cannon2012stochastic} as you can see in Table. \ref{table:time}. On the contrary, our proposed approach requires a long time to compute $\gamma_{t,i} ~(t =0,\dots,T-1,~i = 1,\dots,p)$ offline. The time complexity increases linearly with the length of task horizon. But, this procedure can be computed in advance and \eqref{eq:conservative_way_gamma} can be used for large values of $t$ to reduce the computation time. Thus we can still use the proposed approach in practice.

% \begin{figure}[h]
%     \centering
%     \includegraphics[width=0.86\columnwidth]{average_cost_mt100_e2_iter10_N6.eps}
%     \caption{Comparison for Sum of closed-loop costs over trials}
%     \label{fig:closed_loop_cost}
% \end{figure}
% \begin{figure}[h!]
%     \centering
%     \includegraphics[width=0.9\columnwidth]{multi_init_e2_N6_mt100.eps}
%     \caption{Comparison for Sum of closed-loop costs over multiple initial states}
%     \label{fig:closed_loop_cost_multi_init}
% \end{figure}

\begin{table}[h!]
\centering
    \begin{tabular}{|c|| c c c||} 
 \hline
  & Offline & Algorithm 1 & SMPC of \cite{cannon2012stochastic} \\ [0.5ex] 
 \hline
 Time(s) & 2.2701 & 0.0976 & 0.0911 \\  [1ex] 
 \hline
\end{tabular}
\caption{Computation times: Values are
obtained with a ThinkPad P53, 2.60 GHz Intel Core i7-9850H, 16GB RAM.}
\label{table:time}
\end{table}
\section{Conclusions}
We proposed a novel and efficient approach to design a stochastic MPC for constrained linear systems with an additive disturbance. 
In order to satisfy the imposed state chance constraints, the set of disturbance sequences that our controller needs to be robust against, were constructed \emph{offline} from collected data. We also proposed a novel reformulation strategy of the chance constraints, where the constraint tightening is computed online by adjusting the offline computed sets based on the previously realized disturbances along the trajectory.
The proposed SMPC was recursively feasible and had chance constraints satisfaction in closed-loop with a confidence level. With numerical simulations, we demonstrated that the proposed approach obtains a larger ROA and lower closed-loop costs in average over the existing feasible SMPC.


%%%%%%%%%%%%%%%%%%%%%%%%%%%%%%%%%%%%%%%%%%%%%%%%%%%%%%%%%%%%%%%%%%%%%%%%%%%%%%%%

%%%%%%%%%%%%%%%%%%%%%%%%%%%%%%%%%%%%%%%%%%%%%%%%%%%%%%%%%%%%%%%%%%%%%%%%%%%%%%%%

\bibliographystyle{IEEEtran}
\bibliography{mybib.bib}
%%%%%%%%%%%%%%%%%%%%%%%%%%%%%%%%%%%%%%%%%%%%%%%%%%%%%%%%%%%%%%%%%%%%%%%%%%%%%%%%
\section*{APPENDIX}
\subsection{Details of Parameters of Simulations}
The cost function in \eqref{eq:original_opti} is defined as:
\begin{align*}
 &\ell(\bar{x}_{k}, u_{k}(\bar{x}_{k}))  \\
 &= (\bar{x}_{k}-x_{\mathrm{ref}})^\top Q (\bar{x}_{k}-x_{\mathrm{ref}}) + u_{k}(\bar{x}_{k})^\top R u_{k}(\bar{x}_{k}),\\   &\ell_F(\bar{x}_T) = (\bar{x}_{k}-x_{\mathrm{ref}})^\top Q_F (\bar{x}_{k}-x_{\mathrm{ref}}),
\end{align*}
where $Q=Q_F = I$ and $R = 1$.

\begin{enumerate}[(E1)]
\item \label{param:1st} Dynamics from Iterative SMPC \cite{rosolia2018stochastic}
\begin{align*}
&A = \begin{bmatrix} 1.2 & 1.5\\0 & 1.3 \end{bmatrix}, B = \begin{bmatrix} 0\\1 \end{bmatrix}, H = \begin{bmatrix} 1 & 0\\-1 & 0\\0 &1\\0 &-1 \end{bmatrix}, \\
&h = 20 \times \mathbf{1}_{4\times 1},H_u = \begin{bmatrix} 1&-1 \end{bmatrix}^\top,h_u = \begin{bmatrix} 6&6 \end{bmatrix}^\top.
\end{align*}
\item \label{param:2nd} Dynamics from Strongly Feasible SMPC \cite{korda2011strongly}
\begin{align*}
&A = \begin{bmatrix} 1 & 0\\1 & 1 \end{bmatrix}, B = \begin{bmatrix} 1\\2 \end{bmatrix}, H = \begin{bmatrix} 1 & 0\\-1 & 0\\0 &1\\0 &-1 \end{bmatrix},\\
&h = 20 \times \mathbf{1}_{4\times 1},H_u = \begin{bmatrix} 1&-1\end{bmatrix}^\top ,
h_u = \begin{bmatrix} 2&2 \end{bmatrix}^\top.
\end{align*}
\item \label{param:3rd} Dynamics from Stochastic MPC \cite{cannon2012stochastic} 
\begin{align*}
&A = \begin{bmatrix} 1.6 & 1.1\\-0.7 & 1.2 \end{bmatrix}, B = \begin{bmatrix} 1\\1 \end{bmatrix}, H = \begin{bmatrix} 1 & 0\\-1 & 0\\0 &1\\0 &-1 \end{bmatrix}, \\
&h =10 \times \mathbf{1}_{4\times 1},H_u = \begin{bmatrix} 1&-1 \end{bmatrix}^\top,
h_u = \begin{bmatrix} 10&10 \end{bmatrix}^\top.
\end{align*}
\end{enumerate}

\subsection{Proof of Theorem 1}
Let~\eqref{eq:robust_reformulation_mpc} be feasible at time step $t$ and denote the corresponding optimal auxiliary input $v$ sequence as $[v_{t|t}^\star,v^\star_{t+1|t},\cdots,v^\star_{t+N-1|t}]$, resulting in the corresponding optimal policies  $[u_{t|t}^\star,u^\star_{t+1|t}(\cdot),\cdots,u^\star_{t+N-1|t}(\cdot)]$ with \eqref{eq:aff_pol_in_mpc2}. 
% Since we set $\mathcal{X}_F$ as a robust control invariant set, we have a safe policy $\pi_{\mathrm{safe}}(\cdot)$ satisfying $\mathbb{P}\Big(Ax_{t+N|t}+B \pi_{\mathrm{safe}}(x_{t+N|t}) +w_t \in \mathcal{X}_F  \Big) \geq 1-\alpha$ if $\mathbb{P}\Big(x_{t+N|t} \in \mathcal{X}_F\Big)\geq 1-\alpha$.
Consider a candidate input $v$ sequence: 
\begin{align}\label{eq:cand_pol}
[v_{t+1|t}^\star,\cdots,v^\star_{t+N-1|t},0], 
\end{align}
at the next time $t+1$, resulting the corresponding optimal policies $[u_{t+1|t}^\star(\cdot),\dots,u^\star_{t+N-1|t}(\cdot),Kx_{t+N|t}]$.
% Due to \eqref{eq:input_constr}, each control input from these policies would satisfy the input constraint under any realizations $x_{t+1}$, even with larger disturbance. 
We need to show that sequence~\eqref{eq:cand_pol} is a feasible solution of problem~\eqref{eq:robust_reformulation_mpc} at time step $t+1$, when applying $u^{\star}_{t|t}$ to the system for any $w_t \in \mathbb{W}$. We will show the constraint satisfaction for any $i \in \{1,\dots,p\}$ and divide our proofs into two parts: 
% \begin{itemize}
%     \item 
% \end{itemize}
\begin{enumerate}
    \item State constraint satisfaction of \eqref{eq:state_constr} 
    \item Terminal constraint satisfaction of \eqref{eq:terminal_set}  
\end{enumerate}
\noindent\textbf{1) State Constraint Satisfaction:}

\noindent We will show satisfaction of \eqref{eq:state_constr} for $k=t+N$ at time step $t+1$.
With $A_{\mathrm{cl}}=A+BK$, the left-handed side in \eqref{eq:state_constr} is written as:
\begin{equation}
    \begin{aligned}
    &~[H]_ix_{t+N+1|t+1} \\ =&~[H]_i\{A_{\mathrm{cl}}^{N}\big(A_{\mathrm{cl}}x_{t|t}+Bv_{t|t}^{\star}+w(t)\big) +\\
    &~\sum_{l=0}^{N-1} A_{\mathrm{cl}}^{N-1-l}(Bv_{t+1+l|t+1}+w_{t+1+l|t+1})\} \\
    =&~ [H]_i \{A_{\mathrm{cl}}^{N}\big(A_{\mathrm{cl}}\big(\bar{x}_{t|0}+\sum_{l=0}^{t-1} A_{\mathrm{cl}}^{t-1-l}w(l)\big) + Bv_{t|t}^{\star}+w(t)\big) \\
    &~ +\sum_{l=0}^{N-1} A_{\mathrm{cl}}^{N-1-l}(Bv_{t+1+l|t+1}+w_{t+1+l|t+1})\} \\
    =&~ [H]_i A_{\mathrm{cl}}\bar{x}_{t+N|0} + b^{t,t+N}_{0,i} + a^{t+N}_{t+1,i}
    \end{aligned}
\end{equation}
We have $\bar{x}_{t+N|t} \in \mathcal{X}_F^t$ at the time $t$. Since we have $[H]_i A_{\mathrm{cl}}\bar{x}_{t+N|0} \leq h_i-\gamma_{t+N,i}$ by using $\bar{x}_{t+N|0} = \bar{x}_{t+N|t}-A_{\mathrm{cl}}^N \sum_{j=0}^{t-1} A_{\mathrm{cl}}^{t-1-j} w(j)$, 
\begin{equation}
\begin{aligned}
&~[H]_ix_{t+N+1|t+1} \\
=&~ [H]_i A_{\mathrm{cl}}\bar{x}_{t+N|0} + b^{t,t+N}_{0,i} + a^{t+N}_{t+1,i}\\
\leq &~ h_i-\gamma_{t+N,i}+b^{t,t+N}_{0,i} + a^{t+N}_{t+1,i} \\
\leq &~ h_i ~\forall [w_{t+1|t+1},\dots,w_{t+N
|t+1}] \in \hat{W}^{\mathrm{on},t+N+1}_{t+1 \rightarrow t+N+1,i}  
\end{aligned}
\end{equation}
The last line holds when we consider the definition of $\hat{W}^{\mathrm{on},t+N+1}_{t+1 \rightarrow t+N+1,i} $ in \eqref{eq:varying_w_ineq}. 
Therefore, \eqref{eq:state_constr} for all $i \in \{1,\dots,p\}$ are satisfied for predicted step $k=t+N$ at time step $t+1$.

\noindent\textbf{2) Terminal Constraint Satisfaction:}

\noindent Now we prove that  $\bar{x}_{t+N+1|t+1}$  satisfies the terminal constraint \eqref{eq:terminal_set} at $t+1$ when applying the candidate input sequence \eqref{eq:cand_pol}. 
Considering \eqref{eq:cand_pol} as an input sequence at time step $t+1$ (i.e., $v_{t+N|t+1} = 0$),
\begin{equation}\label{eq:term_pf1}
    \begin{aligned}
    &\bar{x}_{t+N+1|t+1} \\
    =& A_{\mathrm{cl}}^N \bar{x}_{t+1|t+1}+\sum_{l=0}^{N-1} A_{\mathrm{cl}}^{N-1-l}Bv_{t+1+l|t+1}\\
    =&A_{\mathrm{cl}}(A_{\mathrm{cl}}^N\bar{x}_{t|t}\!+\!A_{\mathrm{cl}}^{N\!-\!1}\!Bv_{t|t}^{\star}\!\!+\!\!\sum_{l=0}^{N-2}\! A_{\mathrm{cl}}^{N\!-\!2\!-\!l}Bv_{t\!+\!1\!+\!l|t\!+\!1})\!+\!A_{\mathrm{cl}}^{N}w(t) \\
    =& A_{\mathrm{cl}} \bar{x}_{t+N|t}+A_{\mathrm{cl}}^{N}w(t)
    \end{aligned}
\end{equation}
From \eqref{eq:terminal_set_def}, $\bar{x}_{t+N|t} \in X_F^{t}$ is written as:
\begin{equation}\label{eq:term_pf2}
    \begin{aligned}
    &~\forall l = \{0,\dots,T-N-1-t\}, \forall i = \{1,\dots,p\},\\
    &[H]_i A_{\mathrm{cl}}^{l+1} \bar{x}_{t+N|t} \leq\! h_i -\! \gamma_{N+l+t,i}+[H]_i  \sum_{j=0}^{t-1} A_{\mathrm{cl}}^{N+t+l-j} w(j)        \\
    &\iff  [H]_i A_{\mathrm{cl}}^{l}(A_{\mathrm{cl}}\bar{x}_{t+N|t}+A_{\mathrm{cl}}^{N}w(t)) \\
    &~~~~~~~~\leq  h_i - \gamma_{N+t+l,i}+[H]_i  \sum_{j=0}^{t} A_{\mathrm{cl}}^{N+t+l-j} w(j)  
    \end{aligned} 
\end{equation}
Using \eqref{eq:term_pf1} and \eqref{eq:term_pf2},
\begin{equation}
    \begin{aligned}
    & [H]_i A_{\mathrm{cl}}^{l+1} \bar{x}_{t+N+1|t+1} \\
    & \leq  h_i - \gamma_{N+t+1+l,i}+[H]_i  \sum_{j=0}^{t} A_{\mathrm{cl}}^{N+t+1+l-j} w(j)\\
    & \therefore ~ \bar{x}_{t+N+1|t+1} \in \mathcal{X}_F^{t+1}
    \end{aligned}
\end{equation}
% \begin{equation}\label{eq:term_pf2}
%     \begin{aligned}
%     &L(A_{\mathrm{cl}}^{-1})^t \bar{x}_{t+N|t} \leq g +L(A_{\mathrm{cl}}^{-1})^t \sum_{j=0}^{t-1} A_{\mathrm{cl}}^{N+t-1-j} w(j)  
%     \\
%   \iff &L(A_{\mathrm{cl}}^{-1})^{t+1}(A_{\mathrm{cl}}\bar{x}_{t+N|t}+A_{\mathrm{cl}}^{N}w(t)) \\
%     & \leq g + L(A_{\mathrm{cl}}^{-1})^t \sum_{j=0}^{t} A_{\mathrm{cl}}^{N+t-1-j} w(j) \\
%     & = g + L(A_{\mathrm{cl}}^{-1})^{t+1} \sum_{j=0}^{t} A_{\mathrm{cl}}^{N+t-j} w(j) 
%     \end{aligned} 
% \end{equation}
Therefore the terminal constraint \eqref{eq:terminal_set} is satisfied with the $\bar{x}_{t+N+1|t+1}$ obtained by applying the candidate sequence \eqref{eq:cand_pol} to the system, at time step $t+1$. 

From the state constraint satisfaction  and the terminal constraint satisfaction, the problem \eqref{eq:robust_reformulation_mpc} remains feasible at all time step $t \geq 0$.
% Additionally, we apply the same policy $u^*_{t+1|t}(\cdot)$ to maintain the satisfaction for future times when condition (ii) happens.
$\blacksquare$

\section*{ACKNOWLEDGMENT}
This research work is partially funded by grants ONR-N00014-18-1-2833, and NSF-1931853. 

% \end{thebibliography}

\addtolength{\textheight}{-12cm}  

\end{document}